\documentclass[11pt, letterpaper]{article}
\usepackage{fullpage} 
\usepackage{xspace,xcolor,graphicx,tabularx} 
\usepackage{mathtools,amsthm,amssymb} 
\usepackage{enumitem} 
\usepackage{thm-restate}
\setenumerate[0]{label={\normalfont (\roman*)}}

\usepackage[T1]{fontenc}
\usepackage[UKenglish]{babel}
\usepackage[scaled=0.86]{helvet}
\usepackage{mathptmx}
\usepackage{longfbox}
\DeclareMathAlphabet{\mathcal}{OMS}{ntxm}{m}{n}
\usepackage{dsfont} 
\let\mathbb\relax
\let\mathbb\mathds
\usepackage{stmaryrd} 
\makeatletter
\renewcommand{\paragraph}{%
  \@startsection{paragraph}{4}%
  {\z@}{2.25ex \@plus 1ex \@minus .2ex}{-1em}%
  {\normalfont\normalsize\bfseries}%
}
\makeatother
\usepackage{authblk}

\definecolor{linkblue}{HTML}{001487}
\usepackage[colorlinks=true,allcolors=linkblue]{hyperref}
\usepackage{url}
\usepackage[nameinlink,capitalize,noabbrev]{cleveref}

\newtheorem{theorem}{Theorem}[section]
\newtheorem*{theorem*}{Theorem}
\newtheorem{proposition}[theorem]{Proposition}
\newtheorem{conjecture}[theorem]{Conjecture}
\newtheorem{lemma}[theorem]{Lemma}
\newtheorem{claim}[theorem]{Claim}
\newtheorem{fact}[theorem]{Fact}
\newtheorem{corollary}[theorem]{Corollary}
\theoremstyle{remark}
\newtheorem{remark}[theorem]{Remark}
\theoremstyle{definition}
\newtheorem{definition}[theorem]{Definition}

\numberwithin{equation}{section}
\newcommand\numberthis{\addtocounter{equation}{1}\tag{\theequation}}
\newtheorem{assumption}[theorem]{Assumption}

\newcommand{\setft}[1]{\textnormal{#1}}
\newcommand{\eps}{\epsilon}
\newcommand{\1}{\mathbb{1}}
\newcommand{\id}{\setft{id}}

\newcommand{\N}{\ensuremath{\mathds{N}}}
\newcommand{\R}{\ensuremath{\mathds{R}}}

\newcommand{\bits}{\ensuremath{\{0, 1\}}}

\newcommand{\ot}{\ensuremath{\otimes}}

\newcommand{\tr}[1]{\mathrm{Tr}\!\left[ #1 \right]}
\newcommand{\Tr}{\mathrm{Tr}}

\newcommand{\pr}[1]{\Pr\!\left[ #1 \right]}
\newcommand{\prs}[2]{\Pr_{#1}\!\left[ #2 \right]}

\DeclareMathOperator{\poly}{poly}

\DeclareMathOperator*{\E}{\mathds{E}}

\newcommand{\ket}[1]{|#1\rangle}
\newcommand{\bra}[1]{\langle#1|}

\DeclarePairedDelimiterX\braket[2]{\langle}{\rangle}{#1 \delimsize\vert #2}

\newcommand{\cF}{\ensuremath{\mathcal{F}}}

\DeclareSymbolFont{greekletters}{OML}{ntxmi}{m}{it}
\DeclareMathSymbol{\alpha}{\mathord}{greekletters}{"0B}
\DeclareMathSymbol{\beta}{\mathord}{greekletters}{"0C}
\DeclareMathSymbol{\gamma}{\mathord}{greekletters}{"0D}
\DeclareMathSymbol{\delta}{\mathord}{greekletters}{"0E}
\DeclareMathSymbol{\epsilon}{\mathord}{greekletters}{"0F}
\DeclareMathSymbol{\zeta}{\mathord}{greekletters}{"10}
\DeclareMathSymbol{\eta}{\mathord}{greekletters}{"11}
\DeclareMathSymbol{\theta}{\mathord}{greekletters}{"12}
\DeclareMathSymbol{\iota}{\mathord}{greekletters}{"13}
\DeclareMathSymbol{\kappa}{\mathord}{greekletters}{"14}
\DeclareMathSymbol{\lambda}{\mathord}{greekletters}{"15}
\DeclareMathSymbol{\mu}{\mathord}{greekletters}{"16}
\DeclareMathSymbol{\nu}{\mathord}{greekletters}{"17}
\DeclareMathSymbol{\xi}{\mathord}{greekletters}{"18}
\DeclareMathSymbol{\pi}{\mathord}{greekletters}{"19}
\DeclareMathSymbol{\rho}{\mathord}{greekletters}{"1A}
\DeclareMathSymbol{\sigma}{\mathord}{greekletters}{"1B}
\DeclareMathSymbol{\tau}{\mathord}{greekletters}{"1C}
\DeclareMathSymbol{\upsilon}{\mathord}{greekletters}{"1D}
\DeclareMathSymbol{\phi}{\mathord}{greekletters}{"1E}
\DeclareMathSymbol{\chi}{\mathord}{greekletters}{"1F}
\DeclareMathSymbol{\psi}{\mathord}{greekletters}{"20}
\DeclareMathSymbol{\omega}{\mathord}{greekletters}{"21}
\DeclareMathSymbol{\varepsilon}{\mathord}{greekletters}{"22}
\DeclareMathSymbol{\vartheta}{\mathord}{greekletters}{"23}
\DeclareMathSymbol{\varpi}{\mathord}{greekletters}{"24}
\DeclareMathSymbol{\varrho}{\mathord}{greekletters}{"25}
\DeclareMathSymbol{\varsigma}{\mathord}{greekletters}{"26}
\DeclareMathSymbol{\varphi}{\mathord}{greekletters}{"27}

\newcommand{\reg}[1]{\mathrm{reg}[#1]}

\usepackage{verbatim}

\begin{document}

\title{
Derandomised tensor product gap amplification \\ for quantum Hamiltonians}

\author[1]{Thiago Bergamaschi}
\author[2]{Tony Metger}
\author[3]{Thomas Vidick}
\author[4]{Tina Zhang}

\affil[1]{UC Berkeley}
\affil[2]{ETH Zurich}
\affil[3]{Weizmann Institute and EPFL}
\affil[4]{MIT}

\date{\vspace{-1cm}}

\maketitle

\begin{abstract}
The quantum PCP conjecture asks whether it is $\mathsf{QMA}$-hard to distinguish between high- and low-energy Hamiltonians even when the gap between ``high'' and ``low'' energy is large (constant).
A natural proof strategy is gap amplification: start from the fact that high- and low-energy Hamiltonians are hard to distinguish if the gap is small (inverse polynomial)~\cite{kitaev02} and \emph{amplify} the Hamiltonians to increase the energy gap while preserving hardness.
Such a gap amplification procedure is at the heart of Dinur's proof of the \emph{classical} PCP theorem~\cite{dinur2007pcp}.
In this work, following Dinur's model, we introduce a new quantum gap amplification procedure for Hamiltonians which uses random walks on expander graphs to derandomise (subsample the terms of) the \emph{tensor product amplification} of a Hamiltonian. 
Curiously, our analysis relies on a new technique inspired by quantum de Finetti theorems, which have previously been used to rule out certain approaches to the quantum PCP conjecture \cite{Brando2013ProductstateAT}.
\end{abstract}

\newpage

{
  \hypersetup{linkcolor=linkblue}
  \setcounter{tocdepth}{2}
  \tableofcontents
}

\newpage

\section{Introduction}

The classical PCP theorem \cite{arora92} is one of the landmark achievements of classical complexity theory, and the quantum PCP conjecture \cite{aharonov2013quantumpcpconjecture} is one of the major open problems in quantum complexity theory. The classical PCP theorem was proven for the first time using algebraic techniques, and later reproven using elementary combinatorial techniques in the celebrated 2007 paper of Dinur \cite{dinur2007pcp}. Both approaches have, so far, resisted quantisation. 

The algebraic proof of the PCP theorem seems difficult to quantise because of its reliance on the properties of polynomial codes, which have no clear quantum analogue~\cite{aharonov2018quantuminspiredproofpp}. The combinatorial proof, which pares the approach down to essentials, has seemed more promising as a guide for quantum PCP. Very broadly speaking, Dinur's proof of the PCP theorem consists of the iterated application of two steps:
\begin{definition}[Dinur's template for proving the PCP theorem]
\label{def:dinur's-template}
\quad
\begin{enumerate}
    \item Gap amplification, in which the \textit{promise gap} of a constraint satisfaction problem family is amplified at the cost of blowing up its alphabet size,
    \item Alphabet reduction, in which the locality of the constraint satisfaction problem family is reduced back to a constant, and the gap is shown not to shrink too much in the process.
\end{enumerate}
\end{definition}
Applying these two steps iteratively to an NP-complete constraint satisfaction problem (CSP) with inverse polynomial promise gap yields another NP-complete CSP with constant promise gap.
Thus, given a purported witness for the amplified CSP, a verifier can check that the witness satisfies the amplified CSP (with constant probability of error) by randomly selecting and checking constantly many clauses, which only requires a logarithmic number of random bits. (The number of bits that the verifier uses in this process is referred to as the \emph{randomness complexity} of the verifier.)

Dinur's classical gap amplification procedure (step (i)) starts from a certain `na\"ive' gap amplification procedure, which amplifies the promise gap effectively but blows up the randomness complexity of the verifier. 
To control the randomness complexity of the verifier, Dinur uses expander random walks to derandomise this `na\"ive' procedure and achieve amplification without paying a steep cost in randomness complexity.

In this work, following Dinur's model, we introduce a new quantum gap amplification procedure for Hamiltonians which uses random walks on expander graphs to derandomise (subsample the terms of) the \emph{tensor product amplification} of a Hamiltonian. The tensor product amplification of a Hamiltonian $H$ (normalized such that $0\leq H\leq I$) is simply $I - (I - H)^{\otimes t}$.  It can be easily shown that tensoring $H$ with itself in this way amplifies the promise gap, but this process blows up the number of terms in the amplified Hamiltonian exponentially in $t$ (and consequently also blows up the amount of randomness required to sample a term at random). By contrast, our derandomised version of tensor product amplification achieves gap amplification without increasing the number of terms in the original Hamiltonian by too much. The following is a formal statement of our gap amplification theorem:

\begin{theorem}
[Gap amplification] \label{thm:tp_amplification_intro} Let $H$ be a $k$--local Hamiltonian with $m$ terms, satisfying the conditions in \cref{definition:layered}. Then, its \emph{derandomised $2t$--fold tensor product amplification} $H^{(2t)}$ (formally defined in \cref{def:amplified_ham}) satisfies:

\begin{enumerate}
    \item[1.] $H^{(2t)}$ is $(2t\cdot k)$--local and has $d^{2t} \cdot m$ terms, where $d$ is the degree of some fixed expander graph family.
    \item[2.] \textbf{Completeness.} The lowest eigenvalue of $H^{(2t)}$ is bounded above by 
    \begin{equation}
        \lambda_{\min}(H^{(2t)}) \leq 2 t \times \lambda_{\min}(H)
    \end{equation}
    \item[3.] \textbf{Soundness.} The eigenvalue of $H^{(2t)}$ is bounded below by
    \begin{equation}
        \lambda_{\min}(H^{(2t)})\geq  \min\bigg[\Theta\bigg(\frac{\log t}{t}\bigg), \Theta\bigg(\sqrt{\frac{t}{\log t}} \times \lambda_{\min}(H)\bigg)\bigg],
    \end{equation}
\end{enumerate}
\noindent where the $\Theta$ notation obscures constants dependent on the choice of expander graph family and the original Hamiltonian $H$.
\end{theorem}

To understand \cref{thm:tp_amplification_intro}, 
one should consider the case that $t$ is a constant: in this context, \cref{thm:tp_amplification_intro} says that we can amplify the promise gap of a Hamiltonian family by a constant multiplicative factor, at the cost of increasing both the locality and the number of terms by a constant multiplicative factor. These parameters are comparable, except for one important difference, to those which Dinur achieves with her gap amplification procedure (see \cite[Section 6]{dinur2007pcp}). We elaborate on this difference more thoroughly in \cref{section:related}. 

In order to prove \cref{thm:tp_amplification_intro}, we devise a new technique inspired by \textit{quantum de Finetti theorems} \cite{renner2007symmetry, renner2008security}, which we believe could be of independent interest. Curiously, while quantum de Finetti theorems have previously been used to rule out approaches to the quantum PCP conjecture \cite{Brando2013ProductstateAT}, here we use methods inspired by these statements to argue the soundness of our amplification scheme. We give more context and details about our application of de Finetti theorems in the related work \cref{section:related} below, and in our technical overview (\cref{section:overview}). We dedicate \cref{section:discussion} to a discussion on potential applications of \cref{thm:tp_amplification_intro}.

\subsection{Related work}
\label{section:related}

\paragraph{Quantum gap amplification.} In 2008, Aharonov, Arad, Landau and Vazirani proposed a quantum gap amplification procedure that was also based on Dinur's gap amplification \cite{aalv09}. The key ingredient in their work was the \textit{detectability lemma}, which in some sense states that ``if the ground state energy of some (non-commuting) Hamiltonian $H$ is at least $E$, then sequentially measuring all the terms of $H$ on the ground state should detect at least $\textit{const} \cdot E$ violations (with constant probability)". Although their original motivation lay in the pursuit of the quantum PCP (QPCP) conjecture~\cite{aharonov2013quantumpcpconjecture}, their work subsequently found many applications to quantum many-body physics and quantum information theory, including area laws \cite{ALV12}, $k$-designs \cite{BHH12}, and Gibbs samplers \cite{KB16}.

There are two key differences between the parameters of Dinur's classical amplification procedure and those of the quantum gap amplification procedure proposed by \cite{aalv09}. The first is that (unlike \cref{thm:tp_amplification_intro}) the gap amplification procedure of \cite{aalv09} requires the Hamiltonian family it amplifies to have constant locality to start with.\footnote{Newer versions of the detectability lemma exist which do not require constant locality of the Hamiltonian to which they are applied (see e.g. \cite{AAV16}); however, we do not know of any detectability lemma that both has this property and is sufficiently strong to complete the proof from \cite{aalv09}.} Since \cite{aalv09}'s gap amplification procedure also increases the locality of the Hamiltonian which it amplifies (where the increase in locality is proportional to the increase in promise gap which it achieves, like in \cref{thm:tp_amplification_intro}), this means that \cite{aalv09}'s gap amplification cannot be iterated by itself (without composing it with a locality reduction procedure).

We resolve this issue in this work. Because our gap amplification procedure (\Cref{thm:tp_amplification_intro}) can be iterated by itself, we can achieve applications that, as far as we know, cannot be achieved with the tools developed by \cite{aalv09}. For example, we are able to get a \emph{locality-gap tradeoff theorem} for $\mathsf{QMA}$-complete Hamiltonians:
\begin{theorem}[Locality-gap tradeoff; informal]
\label{thm:locality-gap-tradeoff}
Let $\mathcal{H}$ be a $k$-local Hamiltonian family that is $\mathsf{QMA}$-complete with promise gap $\eps$.\footnote{The members of $\mathcal{H}$ also need to satisfy certain technical conditions, such as being a sum of polynomially many projective terms, which we are suppressing in this informal statement. For a formal version, see \Cref{theorem:iterated_amplification}.} Then, given a deterministic construction of $\mathcal{H}$, there is a deterministic construction of a family of Hamiltonians $\mathcal{H}'$ such that $\mathcal{H}'$ has locality $k \cdot 1/\eps$ and is $\mathsf{QMA}$-complete with some universal constant promise gap $c$.
\end{theorem}
In particular, \Cref{thm:locality-gap-tradeoff} could be applied (for example) to boost the promise gap of a `weak QPCP'. Recent progress \cite{Anshu2023CircuittoHamiltonianFT} suggests potential avenues for achieving a so-called \emph{weak quantum PCP}, namely, a QPCP with $\mathrm{polylog}(n)$ locality and $1/\mathrm{polylog}(n)$ promise gap instead of constant locality and promise gap. Applying \Cref{thm:locality-gap-tradeoff} to such an object would immediately yield a PCP with $\mathrm{polylog}(n)$ locality and \emph{constant} promise gap.

While our gap amplification procedure can be iterated by itself, like Dinur's, there remains a second important difference between the two. Dinur's gap amplification procedure is \emph{non-locality-increasing}, in the sense that one round of amplification blows up the \emph{alphabet size} of the constraint satisfaction problem but not the number of variables that each constraint involves. All attempted quantisations of Dinur's gap amplification, including ours and \cite{aalv09}, are locality-increasing, which is a significant obstacle to applying the kinds of PCP composition theorems that make alphabet reduction possible. 
\cite{aalv09} gap amplification (and ours) is in some sense a quantum version of classical derandomised \emph{sequential repetition}, and in the classical setting derandomised \emph{parallel repetition} seems necessary to make Dinur's template from \cref{def:dinur's-template} work: see \cite{DR06,DM10} for a discussion of this issue. As such, producing a quantisation of Dinur's gap amplification procedure which reproduces all of its important characteristics remains an open problem.

Achieving derandomised sequential (locality-increasing) amplification is already non-trivial in the quantum setting, however, and we view one of our primary contributions as that of analysing a natural but so far unstudied locality-increasing derandomised quantum amplification procedure using techniques that have not previously been applied in this context. As far as we know, our techniques are the first to yield iterability. Moreover, the techniques which \cite{aalv09} used to analyse their own amplification procedure have since found widespread application outside of their original context, and we hope the introduction of a new approach for analysing low-randomness locality-increasing quantum gap amplification might be found to have similar utility. We may have particular grounds for hope because the work of \cite{aalv09} can in some ways be viewed as a \emph{black-box reduction} from the quantum problem to the classical one, while our approach more directly interprets certain classical gap amplification strategies in a quantum setting, and as such may more clearly illuminate the differences between the classical and the quantum problems.

\paragraph{Quantum de Finetti theorems.} A defining property of quantum entanglement is that it is \textit{monogamous}. That is, a quantum system cannot be very entangled with a large number of other systems. Quantum \textit{de Finetti theorems} attempt to quantify this intuition, and more broadly offer a versatile tool to understand quantum correlations, with applications to quantum information theory \cite{renner2007symmetry}, cryptography \cite{Christandl2008PostselectionTF,renner2008security}, algorithms \cite{Brando2013ProductstateAT,ber23} and complexity theory \cite{LW2017,brandao2013quantum}. Roughly speaking, quantum de Finetti theorems capture the fact that a random $k$ qubit marginal of any $n\gg k$ qubit state should be close to a convex combination of product states, and therefore unentangled. 

Most relevant to us are the results of \cite{Brando2013ProductstateAT}, who showed that local Hamiltonians defined on dense (degree $d$) constraint graphs admit product state approximations, in that there exists a product state whose energy is inverse $\poly(d)$ close to the ground state energy.\footnote{Note that the existence of a product state approximation implies that determining the ground state energy up to said approximation error is in $\mathsf{NP}$, since the description of the product state is a succinct classical witness.} In some sense, their results offer a concrete route to \textit{disproving} the QPCP conjecture: if there existed a procedure which increased the degree of any local Hamiltonian, without decreasing
its ground-state energy (or increasing locality), then deciding the ground state energy of the resulting Hamiltonian up to a constant would be in $\mathsf{NP}$ \cite[Corollary 11]{Brando2013ProductstateAT}.

Here, we issue two remarks. First, the conclusions of \cite{Brando2013ProductstateAT} become exponentially weaker if the degree and the locality of the Hamiltonian are allowed to increase simultaneously during amplification \cite{AN22}, which allows us to evade their no-go. Second, despite these ``negative" results, de Finetti techniques will nevertheless play a central role in our approach in understanding when quantum correlations can adversarially decrease the ground state energy of our amplified Hamiltonians. A comprehensive discussion is presented in \cref{section:techniques} and \cref{sec:low-energy-analysis}.

\subsection{Discussion and applications}
\label{section:discussion}

In this section we discuss a number of potential applications of our gap amplification procedure.

\paragraph{Streaming quantum PCP}
\Cref{thm:tp_amplification_intro} can be applied iteratively $1/\log(n)$ times to achieve what we call a `streaming quantum PCP':

\begin{theorem}
    [Streaming quantum PCP]\label{thm:qma_hard_iterative_intro} There exists a deterministic construction of a family $\{H_n\}$ of Hamiltonians on $n$ qubits and an explicit constant $c$, such that each Hamiltonian $H_n$ is a sum of $\poly(n)$ many terms, each summand is an $O(n)$-fold tensor product of $O(1)$-local projections, and it is $\mathsf{QMA}$-complete to decide whether the ground state energy of $\{H_n\}$ is $\leq \mathsf{negl}(n)$ or $\geq c$.
\end{theorem}

\cref{thm:qma_hard_iterative_intro} can be interpreted as a quantum version of what one would get in the classical setting by iterating Dinur's gap amplification procedure from \cref{def:dinur's-template} \emph{without} locality or alphabet reduction reduction. \cref{thm:qma_hard_iterative_intro} can also be viewed as a very preliminary version of quantum PCP. Indeed, the family of Hamiltonians guaranteed by \cref{thm:qma_hard_iterative_intro} satisfies the conditions of the quantum PCP conjecture---including the condition that there are polynomially many terms in each member Hamiltonian---except for the fact that the member Hamiltonians do not have constant locality, which is of course a substantial omission. 

Nonetheless, the terms in each $H_n$ are individually easy to verify: since each term is the $O(n)$-fold tensor product of $O(1)$-local projections, each term can be measured using a constant-depth quantum circuit (and classical post-processing), or by a quantum algorithm that reads a constant number of qubits at a time. There are also only polynomially many such terms in each $H_n$, which implies that the energy of $H_n$ can be measured on a witness state by a verifier who samples a random term and measures it, and in the process uses only $O(\log n)$ many classical random bits and very limited quantum resources. (We call \cref{thm:qma_hard_iterative_intro} a `streaming quantum PCP' for this reason.) It also implies that time evolution under $H_n$ can be efficiently simulated (see \cref{sec:iterated_to_constant}). Ease of verification is one of the chief motivations for quantum PCP, and ease of simulation is one of the chief motivations for the \emph{sparse} quantum PCP conjecture (see below). As such, \cref{thm:qma_hard_iterative_intro} can be viewed as progress towards both objectives.

\begin{remark}
There is an alternative way to achieve \Cref{thm:qma_hard_iterative_intro} which does not use our \Cref{thm:tp_amplification_intro} and instead uses improvements to the techniques of \cite{aalv09} due to \cite{AAV16} to prove \Cref{thm:qma_hard_iterative_intro} using ``one-shot'' (non-iterative) amplification. This result, due to Anurag Anshu, is presented in \cref{section:dl}. The main drawback of this alternative approach is that it does not allow the kind of more ``fine-grained'' amplification that our \Cref{thm:tp_amplification_intro} provides, where the tunable parameter $t$ simultaneously governs the gap amplification and the locality increase of the input Hamiltonian. 
\end{remark}

\paragraph{Locality-gap tradeoffs for local Hamiltonians.} As mentioned above in \cref{thm:locality-gap-tradeoff}, \cref{thm:tp_amplification_intro} can be iteratively applied in order to produce a `locality-gap tradeoff' for local Hamiltonians (see \cref{theorem:iterated_amplification} for the general statement). More explicitly, \cref{thm:tp_amplification_intro} implies that, if there is a deterministic construction of a $\mathsf{QMA}$-complete family of Hamiltonians with locality $\ell$ and promise gap $\eps$, then (under mild conditions) there is a deterministic construction of a $\mathsf{QMA}$-complete family of Hamiltonians with locality $\ell /\eps$ and gap which is some universal constant $c$. 

In particular, we could use \cref{thm:tp_amplification_intro} to boost a `weak quantum PCP' statement (cf.~\cite[Open Problems]{Anshu2023CircuittoHamiltonianFT})---that is, a family of Hamiltonians with $\poly\log n$ locality and inverse $\poly\log n$ promise gap---into a family of Hamiltonians that has $\poly\log n$ locality and \emph{constant} promise gap.

\paragraph{Quantum games PCP.}

The quantum games PCP conjecture states that there exists an $\mathsf{MIP}^*$ protocol (i.e., a protocol between a classical verifier and multiple non-communicating but potentially entangled provers) with $\poly \log n$ sized questions and answers, constant completeness-soundness gap, and efficient provers that can decide all of $\mathsf{QMA}$ \cite{natarajan2024statusquantumpcpconjecture}. The motivation for this question lies in the connection between PCPs and succinct $\mathsf{MIP}$ (classical multi-prover) protocols with constant gap in the classical world  \cite{aharonov2013quantumpcpconjecture}. The connection between quantum PCP and succinct $\mathsf{MIP}^*$ protocols is less immediate, although, like in the classical case, quantum games PCP is a necessary consequence of Hamiltonian quantum PCP (but this is considerably harder to show than it is classically). A proof of the quantum games PCP conjecture was claimed by \cite{NV18}, but it has since been established that the proof had errors, and as of now the conjecture remains open.

We believe \cref{thm:qma_hard_iterative_intro} implies a deterministic construction of an $\mathsf{MIP}^*$ protocol with efficient provers, constant completeness-soundness gap and $\log n$ sized \emph{questions} ($\poly n$ sized answers) that decides all of $\mathsf{QMA}$. Such a protocol was not known prior to our work, and it would follow by designing a succinct Pauli braiding test \cite{dlS, MNZ24} that supports measurements of the terms in our amplified Hamiltonian, which could be made into tensor products of 5-local Clifford projections \cite{broadbent2016zero}. Unfortunately, our result does not appear to help with achieving succinct answers, and we leave as an open question whether it is possible to build on these results to prove the quantum games PCP conjecture.

\paragraph{Sparse quantum PCP.}
The sparse quantum PCP conjecture is similar to the original quantum PCP conjecture except that, instead of requiring that the Hamiltonian is a sum of local terms, we require that the Hamiltonian is expressible as a sparse matrix with efficiently computable entries. 

This conjecture is of interest primarily for two reasons: 
\begin{enumerate}
    \item[1.] It is, like the NLTS theorem \cite{ABN2023}, a necessary consequence of quantum PCP, and
    \item[2.] Sparse Hamiltonians are simulable given only oracle access to their entries, which allows their energy to be measured even in the absence of any term decomposition.
\end{enumerate}

\noindent This entails that their ground state energies are easy to verify, and as such a sparse quantum PCP can be viewed as a meaningful stepping stone towards the kind of easy verification that full quantum PCP requires. 

One might hope that \cref{thm:qma_hard_iterative_intro} could be useful for sparse QPCP, because our Hamiltonian is a sum of polynomially many relatively structured terms (see below, \cref{def:amplified_ham}). Unfortunately, existing black-box approaches to matrix sparsification appear to be either too randomness-expensive or too computationally inefficient to yield the sparse QPCP conjecture when applied to the terms of our Hamiltonians. We leave it as an interesting open question whether progress can be made in this direction by exploiting the particular tensor-product structure of our Hamiltonians.

\subsection{Technical overview}
\label{section:overview}

\subsubsection{Construction and notation}

In this subsection we more formally present our construction of a quantum gap amplification procedure. Our starting point is a Hamiltonian $H$ on $n$ qubits, which can be partitioned into a convex combination of $g$ commuting layers (or colors).

\begin{definition}
    [Layered Hamiltonians]\label{definition:layered} $H$ is said to be a \emph{layered Hamiltonian} if it admits a decomposition
    
    \begin{equation} \label{equation:H_definition}
    H = \sum_{\chi\in [g]} w_\chi \cdot H_\chi\;, \quad \text{ where each layer }\quad H_\chi = \E_{i\in [m_\chi]} \Pi^\chi_i \end{equation}
    \noindent  is the expectation over $m_\chi$ commuting projections, and the weights $\{w_\chi\}_{\chi\in [g]}$ are positive and $\sum_\chi w_\chi=1$. To quantify imbalance between the layers, we introduce the parameter 
\begin{equation}
    \omega_{\min} \equiv (\min_\chi w_\chi)^{-1}.
    \label{eq:min-weight-def}
\end{equation}
\end{definition}

 As we discuss later on (\cref{section:qma_complete_equitable}), there are $\mathsf{QMA}$-complete local Hamiltonians which admit such a decomposition with a constant number $g$ of layers (and constant $\omega_{\min}$), so starting from such a decomposition is essentially without loss of generality. Our construction will amplify the individual commuting layers separately. To do so, for each color we impose an expander graph structure on the clauses $m_\chi$.\footnote{So long as $m_\chi$ is above some constant $m_0$, there exists a determinstic construction of such a graph, see \cref{lemma:expander_construction} \cite{Reingold2000EntropyWT}. The case $m_\chi\leq m_0$ is addressed by trivially repeating the clauses in $H_\chi$. }

\begin{definition}
    [Paths of length $t$ in $G_\chi$]\label{definition:paths} Let $G_\chi$ be a $d$-regular $\lambda$-spectral expander graph on $m_\chi$ vertices (\cref{definition:spectral}). Define a function family $\mathcal{F}_\chi:\{f:[t]\rightarrow [m_\chi]\}$ such that there is a one-to-one\footnote{$\forall f \in \mathcal{F}_\chi$, the tuple $(f(1), \dots, f(t)) \in [m_\chi]^{t}$ is a path of length $t$ in $G_\chi$, and every such path is represented by some $f \in \mathcal{F}_\chi$.} correspondence between member functions $f \in \mathcal{F}_\chi$ and paths of length $t$ in $G_\chi$. 
\end{definition}

\noindent The amplified Hamiltonian will act on $t$ (tensor) copies of the original $n$ qubit system. For each color $\chi$ and path $f\in \mathcal{F}_\chi$, we associate to $f$ a clause $\Pi_f^\chi$ in the amplified Hamiltonian, in the form of a projection onto the intersection of the $t$ clauses on the path. In other words, the clauses in the amplified Hamiltonian are satisfied if all (the ``AND") of the clauses on the path $f$ are satisfied:
\begin{equation}
    \Pi_f^\chi \equiv \mathbb{I} - \bigotimes_{j\in [t]} \big(\mathbb{I}-\Pi^\chi_{f(j)}\big)
\end{equation}

\begin{remark}
    It is easy to check this operator is a projection: it is positive and squares to itself. 
\end{remark}

The $t$-fold amplified version of $H$, which we denote by $H^{(t)}$, can then be expressed as follows:

\begin{definition}
[Derandomised Tensor Product Amplification]\label{def:amplified_ham} Let $H$ be layered as in \cref{definition:layered}. We define the $t$-fold amplification $H^{(t)}$ of $H$ with respect to $\{\mathcal{F}_\chi\}_{\chi\in [g]}$ as
\begin{equation}
        H^{(t)} = \sum_{\chi \in [g]} w_\chi\cdot H^{(t)}_\chi,\quad  \text{where} \quad H^{(t)}_\chi\equiv \E_{f \in \cF_\chi}\bigg(\mathbb{I} - \bigotimes_{j\in [t]} \big(\mathbb{I}-\Pi^\chi_{f(j)}\big)\bigg).
    \end{equation}
\end{definition}

\noindent If the original $H$ was a $k$-local Hamiltonian on $n$ qubits and $m$ clauses, then $H^{(t)}$ is $(k\cdot t)$-local Hamiltonian on $n\cdot t$ qubits and $d^t\cdot m$ clauses. 

\begin{remark}
    Our construction of the amplified Hamiltonian is similar to that of \cite{aalv09}, but we use tensor products instead of matrix products. 
\end{remark}

\subsubsection{Main techniques}
\label{section:techniques}

Here we highlight the main techniques behind the proof of \cref{thm:tp_amplification_intro}.

\paragraph{Background: tensor product amplification.}
The starting point for our gap amplification procedure is tensor product amplification. The $t$-fold tensor product amplification of a Hamiltonian $H$ is simply
\begin{equation}
\mathbb{I} - (\mathbb{I}  - H)^{\otimes t}.
\end{equation}
If $H$ has a decomposition as a sum of projections,
$H = \E_{i \in [m]} \Pi_i$,
then the tensor product amplification of $H$ can be expanded as

\begin{equation}
\E_{(i_1, \dots, i_t) \in [m]^t} \Big(\mathbb{I}  - \bigotimes_{j \in [t]} (\mathbb{I}  - \Pi_{i_j}) \Big). \label{eq:tensor-amp-decomposition}
\end{equation}
If $H$ had $m$ terms, therefore, the $t$-fold tensor product amplification of $H$ has $m^t$ terms.

Analysing tensor product amplification is straightforward: one simply has to observe that, since the operator $\mathbb{I}  - (\mathbb{I} -H)^{\otimes t}$ is diagonal in the basis constituted by tensor products of eigenvectors of $H$, the ground energy of $\mathbb{I}  - (\mathbb{I} -H)^{\otimes t}$ is attained by a tensor product across $t$ registers of ground states of $H$. If we start with a Hamiltonian family $\{H_n\}_n$ with promise gap of size $\gamma$ (and completeness sufficiently close to 1), then the gap increases at least to $t\gamma$ after $t$-fold tensor product amplification.

We would, however, like to achieve a similar amount of gap amplification without increasing the number of terms so steeply. In particular, instead of $m^t$ terms, we would like the $t$-fold amplification of $H$ to have $c(t) \cdot m$ terms, where $c(t)$ is a function that depends arbitrarily on $t$ but not on $m$. Then we can set $t$ to be a constant and iterate gap amplification $O(\log n)$ times to get a new Hamiltonian family $\{H'_n\}_n$ such that $H'_n$ only has a multiplicative factor $\poly(n)$ more terms than $H_n$, but the promise gap of $\{H'_n\}_n$ is also a $\poly(n)$ factor larger than that of $\{H_n\}_n$.

\begin{remark}
If we are content to show QMA-completeness of streaming quantum PCPs under \emph{randomised reductions}, one way that we could reduce the number of terms in \cref{eq:tensor-amp-decomposition} is to pick some to keep at random and delete the others. More specifically, we could prove an analogue of our \cref{thm:qma_hard_iterative_intro} under randomised reductions by randomly subsampling terms of \cref{eq:tensor-amp-decomposition} using the matrix Chernoff bound: see e.g. \cite{natarajan2024statusquantumpcpconjecture}. The gap amplification statement that we prove, \cref{thm:tp_amplification_intro}, can be viewed as the `correct' derandomisation of matrix-Chernoff-based subsampling of tensor product amplification. \cref{thm:tp_amplification_intro} allows us to show \cref{thm:qma_hard_iterative_intro} under deterministic reductions. 

One advantage of derandomisation in view of applications is that \cref{thm:tp_amplification_intro} is \emph{iterable}, whereas the matrix-Chernoff-based approach is not. That is, \cref{thm:tp_amplification_intro} holds even when $t$ is a constant, but matrix-Chernoff can only be used to do `one-shot' amplification (to amplify in one go up to constant gap), since it requires the Hamiltonian family being subsampled to have a constant promise gap.
\end{remark}

\paragraph{Background: classical sequential repetition and its derandomisation}
The classical analogue of tensor product amplification is \emph{sequential repetition}, in which, given a constraint satsifaction problem (CSP) $C$ with $m$ clauses $c_1, \dots, c_m$, one constructs a new CSP $C'$ whose clauses are the ANDs of clauses from $C$: that is, every clause in $C'$ is of the form $c_{i_1} \land \dots \land c_{i_t}$ for $(i_1, \dots, i_t) \in [m]^t$. If $C$ is satisfiable, than $C'$ is also satisfiable; and, if $\gamma$ fraction of clauses in $C$ are unsatisfied by the most satisfying possible assignment to $C$ (the `unsatisfiability' of $C$ is at least $\gamma$), then this fraction will be at least $t \gamma$ for $C'$.

In the classical case, we can also consider \emph{derandomised} sequential repetition, which achieves a similar amplification guarantee but using far fewer clauses. We can achieve derandomised sequential repetition by first constructing $C'$, the full sequential repetition of $C$, and then \emph{subsampling} the clauses of $C'$ using random walks on expander graphs. More specifically,

\begin{definition}[Derandomised sequential repetition]
\label{def:derandomised-seq-rep}
    Given a (classical) CSP $C$ on $n$ bits and $m$ clauses, its $t$-fold derandomised sequential repetition is a CSP $C^{(t)}$ on $t\cdot n$ bits where the new clauses are the AND of all the clauses that lie along any length $t$ path in some expander graph defined over $[m]$. 
\end{definition}
So that this exposition is self-contained, we will sketch how derandomised sequential repetition can be analysed using a particular property of random walks on expander graphs \cite[Chapter 4]{vadhan-pseudorandomness} \cite[Section 6]{dinur2007pcp}. The only property of random walks which we need is a `set-avoiding' property: For any fixed sets $A, B \subseteq [m]$ of no more than $\delta \cdot m$ in size, the probability that a random walk that starts in $A$ ends up in $B$ after $s$ steps is bounded above by $\delta + \mu^{s}$, where $\mu$ is the second largest eigenvalue of the expander graph. This property follows in a straightforward way from \cref{lemma:expander-quadratic}.

We sketch how this property would be useful for a classical analysis of derandomised sequential repetition. Recall that $C^{(t)}$ is the derandomised sequential repetition of $C$, as defined in \cref{def:derandomised-seq-rep}. We define a quantity called the \emph{violation number} which will be useful in the analysis:
\begin{definition}[Violation number]
     The violation number $X_{x, C^{(t)}}$ is a random variable that can be sampled by picking a uniformly random `path clause' $p = c_1 \land \cdots \land c_t$ from $C^{(t)}$ and evaluating how many out of the $t$ clauses in the AND is violated by $x$.
\end{definition}
The violation number is a non-negative integer, and $\mathrm{Pr}_p[X_{x, C^{(t)}} > 0]$ precisely captures the fraction of clauses of $C^{(t)}$ which are violated by $x$. Recall the Second Moment Method, which for any random variable $X$ relates $\mathrm{Pr}[X>0]$ to the mean and variance of $X$:
\begin{fact}
    [Second Moment Method]
    For any non-negative random variable $X\geq 0$, 
    \begin{equation}
        \mathbb{P}[X>0]\geq \frac{\mathbb{E}[X]^2}{\mathbb{E}[X^2]}
    \end{equation}
\end{fact}
As such, if we can get a \emph{lower bound} on the expectation of $X_{x, C^{(t)}}$ and an \emph{upper bound} on the expectation of its square, we will obtain a lower bound on $\mathrm{Pr}[X_{x, C^{(t)}} > 0]$ and by extension a lower bound on the unsatisfiability of $C^{(t)}$.

The expectation of $X_{x, C^{(t)}}$ is easy to control using linearity of expectation. We now sketch how to control the expectation of the square using the `set-avoiding' property of random walks.
For any classical assignment $x\in \{0, 1\}^{nt}$ to $C^{(t)}$, $x$ can be grouped into blocks $x_1, \cdots, x_t\in \{0,1\}^n$---which in some sense are assignments to the original $C$---and one can define a set of clauses $\mathsf{Viol}_{x, i}\subset [m]$ which contain the clauses of $C$ violated by $x_i$. Then we see that the expectation of the square of the violation number is precisely
\begin{align}
&\sum_{i,j \in [t]} \mathrm{Pr}_{p \leftarrow C^{(t)}}[c_i \in \mathsf{Viol}_{x,i} \text{ and } c_j \in \mathsf{Viol}_{x,j}] \\
&= \sum_{i,j \in [t]} \mathrm{Pr}_{p \leftarrow C^{(t)}}[c_i \in \mathsf{Viol}_{x,i}] \cdot \mathrm{Pr}_{p \leftarrow C^{(t)}}[c_j \in \mathsf{Viol}_{x,j} \:|\: c_i \in \mathsf{Viol}_{x,i}]
\label{eq:classical-violation-square}
\end{align}
where the probability is over the sampling of the random `path clause' $p \coloneqq c_1 \land \cdots \land c_t$. The probability $\mathrm{Pr}_{p \leftarrow C^{(t)}}[c_j \in \mathsf{Viol}_{x,j} \:|\: c_i \in \mathsf{Viol}_{x,i}]$ is exactly the probability that a random walk which starts in $\mathsf{Viol}_{x,i}$ ends up in $\mathsf{Viol}_{x,j}$ after $|j-i|$ steps, which is the type of quantity controlled by the `set avoiding' property.

\begin{remark}
Since our proof strategy revolves around the second moment of $X_{x,C^{(t)}}$, derandomised sequential repetition would be even easier to analyse if the `path clauses' $p \leftarrow C^{(t)}$ were sampled \emph{pairwise independently} instead of using a random walk. Pairwise independent sampling is too randomness-expensive to yield the parameters we want, so we cannot use it in \cref{def:derandomised-seq-rep}; however, this observation is the motivation for the proof strategy that we adopt (see `Technique 1: commuting layers' below).
\end{remark}

\paragraph{Difficulties in quantisation}
Although this outline is sufficient for classical CSPs, unfortunately, we encounter an obstacle when we try to quantise the argument above. We can, analogously with \cref{def:derandomised-seq-rep}, define the \emph{derandomised tensor product amplification} of a Hamiltonian $H$ and call it $H^{(t)}$: this is the motivation for \cref{def:amplified_ham}, in which we define our construction of the amplification of a Hamiltonian. However, $H^{(t)}$ may be non-commuting, so there need not be a valid definition of the violation sets $\mathsf{Viol}_{x, i}$ which we relied on in the classical analysis. Moreover, it may have exclusively entangled ground states, so measuring a clause could collapse the state and affect the probabilities that other clauses are violated.

In this overview we will not attempt to give a complete sketch of our analysis: the reader who wants to understand our entire proof can go directly to the technical sections, which contain further exposition. Here we will focus instead on describing at a high level the two main techniques that we use to overcome the difficulty that the ground states of $H^{(t)}$ may be entangled. Both these techniques rely in a sense on reducing the analysis of an entangled state to that of a convex combination of product states, albeit in different ways.

\paragraph{Technique 1: commuting layers}
If all the terms in the original Hamiltonian $H$ \emph{commute}, then we can reduce our analysis to the classical case: we simply imagine measuring the potentially entangled ground state of $H^{(t)}$ in a complete basis that diagonalises all the terms in $H^{\otimes t}$ simultaneously. This collapses the state to a mixture over product states, which we can deal with using the classical argument and convexity. Of course, it is not true that all the terms in $H$ commute, but we can conduct some parts of our analysis (e.g.~\cref{lemma:second_moment_Xs}) by partitioning the terms in $H$ into at most constantly many commuting \emph{layers}, analysing each layer separately, and recombining the layers afterwards (up to some amount of loss).\footnote{Note that, since our amplification procedure (see \cref{def:amplified_ham}) preserves the number of commuting layers, we can iterate our procedure as many times as we like as long as we start with an $H$ that has constantly many layers. We show that such Hamiltonians are QMA-complete in \cref{section:qma_complete_equitable}.}

We use this commuting layers technique only in order to prove that random walk subsampling of terms can be approximately by \emph{pairwise independent} subsampling in our setting. We then develop a new technique---our `miser's de Finetti' technique, which we introduce shortly---in order to prove that pairwise independent sampling is similar enough to fully independent sampling for us even on entangled states. Fully independent sampling is equivalent to (non-derandomised) tensor product amplification, so this statement allows us to get good bounds on the amplification guarantees of our procedure. This last step -- i.e. relating the pairwise independent case to the fully independent case -- is arguably the hardest part of our analysis.

\begin{remark}
The authors of \cite{aalv09} encounter similar issues, and also resolve them by partitioning the original Hamiltonian into commuting layers. However, they carry through their \emph{entire} analysis using the commuting-layer approach to deal with the entanglement problem, while we use both the commuting-layer approach and also our `miser's de Finetti' technique. Their approach gives rise to constants in the amplification bounds that depend on the locality of terms in $H$, and in particular become unmanageable when the locality of $H$ is larger than constant; ours, perhaps surprisingly, does not. This is the chief reason we are able to iterate our procedure.
\end{remark}

\paragraph{Technique 2: miser's de Finetti}
Quantum information theory gives us a class of tools for reducing entangled states to convex combinations of product states, in the form of the \emph{quantum de Finetti theorems} \cite{renner2007symmetry,renner2008security}. There are many types of quantum de Finetti theorem, but they all say something like the following: if we take a \emph{symmetric} (permutation-invariant) pure quantum state $\rho$ on $t$ registers and trace out all but $k\ll t$ registers, the mixed state left over on those registers is `close' to a convex combination of product states.

Na\"ively, we might try to use a de Finetti theorem to solve our problem as follows. Recall that each term $\Pi_{\mathrm{amp}}$ in the amplified Hamiltonian $H^{(t)}$ (\cref{def:amplified_ham}) is a tensor product $\Pi_{\mathrm{amp}} = (\1-\Pi_1) \otimes \cdots \otimes (\1-\Pi_{t})$, where $\Pi_1, \dots, \Pi_t$ are terms in $H$, and the constituent terms of $\Pi_{\mathrm{amp}}$ as well as the order in which they appear are determined by a walk. Instead of placing the $t$ terms in $\Pi_{\mathrm{amp}}$ in the order that the walk dictates, we could pick a random permutation $\pi$ and set $\Pi_{\mathrm{amp}} = (\1-\Pi_{\pi(1)}) \otimes \cdots \otimes (\1-\Pi_{\pi(t)})$ instead---or, if we want our reduction to remain deterministic, we could put all $t!$ possible orderings as terms. This is equivalent to randomly permuting the registers of the state $\tau$ to which we apply the terms, so this manoeuvre allows us to assume that $\tau$ is permutation symmetric. We trace out all but the last $k\ll t$ registers of $\tau$, leaving us with a state which is (mostly) a convex combination of (approximately) product states. We can then apply the classical argument to these $k$ remaining registers.

Unfortunately, this na\"ive plan does not work for quantitative reasons. The best possible de Finetti theorem requires $k = \Omega(\log d)$, where $d$ is the dimension of the Hilbert space associated with a single register. This would mean that we get no amplification unless we set $t \geq \poly(n)$, which makes random walk sampling untenably expensive. This would also mean that, if we want a deterministic reduction, we need to have order $n!$ terms in the amplified Hamiltonian. Either of these facts is a sufficient obstacle to achieving an amplified Hamiltonian with a polynomial number of terms.

As such, we cannot hope to rely on a quantum de Finetti theorem as a black box in our setting.
However, perhaps surprisingly, we are able to make progress with a technique that is inspired by a particular style of proof for a de Finetti theorem. A strategy for proving a de Finetti theorem goes as follows \cite{vidick2016simple}. The keystone of the strategy is the following fact, which can be made quantitative and then proven using elementary techniques:
\begin{claim}
    If $\rho$ is any permutation symmetric state on $t$ registers, and we \emph{project} the last $t-k$ registers onto the state $\ket{\psi}^{\otimes (t-k)}$ for some pure single-register state $\ket{\psi}$, then the first $k$ registers of the residual state must be `close' to the pure state $\ket{\psi}^{\otimes k}$, because the state was permutation symmetric.
    \label{claim:dF-keystone}
\end{claim}
It then turns out that, for any permutation symmetric state $\rho$ on $t$ registers, we can write $\rho$ with the last $t-k$ registers traced out as a convex combination of states of the form $(\id_{t-k} \otimes \bra{\psi}^{\otimes k}) \rho (\id_{t-k} \otimes \ket{\psi}^{\otimes k})$. Combining these two deceptively simple observations yields the de Finetti theorem: if we take a symmetric state $\rho$ on $t$ registers and trace out the final $t-k$ registers, the mixed state left over on the first $k$ registers is `close' to a convex combination of product states $\ket{\psi}^{\otimes k}$.

This argument is at the heart of \cite{vidick2016simple}, which gives an astonishingly simple proof of the so-called `exponential de Finetti theorem'. However, in order to make \cref{claim:dF-keystone} true quantitatively, $t$ has to be quite large. Intuitively, we can understand this as follows: arguing that the first $k$ registers of a permutation-symmetric state $\rho$ must look similar to $\ket{\psi}^{\otimes k}$ if the last $t-k$ registers are in the state $\ket{\psi}^{\otimes (t-k)}$ is rather like collecting \emph{measurement statistics} about $\rho$ from the last $t-k$ registers and then arguing that, since the state is permutation symmetric, the measurement statistics are representative of the first $k$ registers as well. If $t$ is small, the statistics which were collected from the measured registers are unlikely to be informative about the remaining $k$ registers. In particular, requiring the remaining $k$ registers to be in a state close to $\ket{\psi}^{\otimes k}$ is a stringent condition---we are characterising the leftover state very finely if we find that this condition is true---and one would therefore expect to look at a large number of registers $t$ before one could conclude such a thing with any reasonable probability.

The insight which allows us to prove \cref{lem:2ndmoment} is that, in our situation, we don't actually \emph{need} such a fine characterisation of the state. Indeed, we don't care if the state $\tau$ is actually product or not: much weaker guarantees suffice for us. As such, we can potentially get away with collecting far fewer `measurement statistics' about $\tau$ than we would have to collect if we wanted to know that $\tau$ was product. For example, it is sufficient for us that $\tau$ `looks product' with respect to measurements of $H$, the original Hamiltonian  (because the quantity we are trying to bound involves only the trace of $H$ on $\tau$ in various registers). As such, we could measure $\tau$ in an eigenbasis that diagonalises $H^{\otimes t}$, and then use a \emph{classical} de Finetti theorem on the measurement outcomes. It turns out that this approach also requires an unviably large $t$ (since classical de Finetti theorems also scale unfavourably with the alphabet size of the random variables)---but we can do even better by realising that we don't need to know the precise eigenvalues of the eigenstates we get in each register after the eigenbasis measurement. Indeed, a very coarse approximation (`is the eigenvalue big or small?') will suffice. This coarse-graining, which allows us to reduce the alphabet size of the measurement outcomes to which we apply de Finetti--inspired techniques, is the chief reason that we can reduce $t$ to something manageable.

This is the idea behind the auxiliary measurement we introduce in \cref{section:aux_energy_meas}. The binary projective measurements $\{\Pi^{< \alpha},\mathbb{I}-\Pi^{< \alpha}\}$ (\cref{definition:aux_meas}) act on the same Hilbert space as $H$, and $\Pi^{< \alpha}$ simply projects onto the energy eigenspaces of $H$ with lower energy than $\alpha$: therefore, $\Pi^{<\alpha}$ commutes with $H$. Measuring $\Pi^{< \alpha}$ on a random set of $t/2$ registers allows us to estimate the ``energy" $E$ of $\tau$ (the estimate we use is the number of $\geq \alpha$ outcomes). We then  measure the complementary $t/2$ registers, and use Chernoff-bound-like tools to argue that they are likely to land in eigenspaces with similar energy $E$. After that, our proof proceeds via a careful case analysis, which hinges on whether $E$ is `high' or `low'. In effect, our strategy is to partition the Hilbert space in which $\tau$ lives into (constantly many) subspaces, each of which has predictable behaviour with respect to $H$; and we are able to proceed with the analysis after we project (‘pinch’) $\tau$ into one of these subspaces because the ‘pinching’ measurement commutes with the measurements of $H$, which are the only measurements that we care about in \cref{lem:2ndmoment}.

\subsection{Organization}

We begin in section \cref{sec:preliminaries} by presenting basic definitions of local Hamiltonians and expander graphs, as well as some basic probability facts. 

We proceed in \cref{section:viol_number_measurement} by formally introducing the violation number measurement, an observable which quantifies the number of violated clauses on a random length $t$ path. In \cref{section:completeness} we argue the completeness of the amplification scheme (\cref{thm:tp_amplification_intro}, Completeness), by applying the first moment method to the amplified Hamiltonian.

In \cref{section:aux_energy_meas} we present the auxiliary energy measurement, a central technical tool which captures the de Finetti ingredient of our proof. In \cref{sec:soundness}, we present the proof of soundness of our amplification scheme (\cref{thm:tp_amplification_intro}, Soundness), wherein we apply the second moment method to the violation number measurement. 

In \cref{sec:iterated_to_constant}, we present the proof of \cref{thm:qma_hard_iterative_intro}, via the iterated application of \cref{thm:tp_amplification_intro}. In \cref{section:qma_complete_equitable}, we describe the base case of our iteration, on local projection Hamiltonians with a constant number of layers.

\paragraph{Acknowledgements}
We thank Anurag Anshu for useful discussions and suggesting the DL-based construction of a streaming quantum PCP presented in \cref{section:dl}.
We thank Dorit Aharonov, Niko Breuckmann, Sevag Gharibian, Louis Golowich, Anand Natarajan, Quynh The Nguyen, Nikhil Srivastava, and John Wright for helpful discussions.
Part of this work was completed while the authors were visiting the Simons Institute for the Theory of Computing.
TM acknowledges support from SNSF Grant No.~20QU-1\_225171 and NCCR SwissMAP. TV is supported by AFOSR Grant No. FA9550-22-1-0391 and ERC Consolidator Grant VerNisQDevS (101086733).
TZ was partially supported by NSF CAREER grant number 2339948.

\section{Preliminaries}
\label{sec:preliminaries}

We dedicate this section to basic facts on the local Hamiltonian problem, probability theory, and expander graphs.

\subsection{The Local Hamiltonian problem}
\label{sec:local_hamiltonians}

The central problem of study in Hamiltonian complexity theory is computing the ground state energy of a local Hamiltonian.

\begin{definition}
    [$k$-$\mathsf{{LH}[a, b]}$] Let $H = \sum_i^m h_i$ be a $k$-local Hamiltonian on $n$ qubits and $m = \poly(n)$ terms, where each $h_i$ is a hermitian operator acting on $k$ qubits and is described using $\poly(n)$ bits. Let $0\leq H\leq \mathbb{I}$.
    
    Further, let $a < b$ be two real parameters described using $\poly(n)$ bits; $b-a$ is referred to as the ``promise gap" of $H$. The $k$-local Hamiltonian problem $k$-$\mathsf{{LH}[a, b]}$ then consists of the task of deciding whether the ground state energy $\lambda_{min}(H)\leq a$ or $\lambda_{min}(H)\geq b$, promised that $H$ satisfies one of the two cases. 
\end{definition}

\cite{kitaev02} proved that the local Hamiltonian problem is $\mathsf{QMA}$-Complete. 

\begin{theorem}
    [\cite{kitaev02}] The $5$-$\mathsf{{LH}[2^{-\poly(n)}, 1/\poly(n)]}$ problem is $\mathsf{QMA}$-complete. 
\end{theorem}

The Quantum PCP Conjecture stipulates that the local Hamiltonian problem remains just as hard under a constant promise gap. 

\begin{conjecture}
    [The Quantum PCP Conjecture \cite{aharonov2013quantumpcpconjecture}] There exists a constant locality $k$ as well as $a, b\in [0, 1]$ satisfying $b-a=\Omega(1)$ such that the $k$-$\mathsf{{LH}[a, b]}$ problem is $\mathsf{QMA}$-complete.
\end{conjecture}

\subsection{Probability}
\label{sec:probability}

Here we recollect a series of simple inequalities. 

\begin{fact}
    [Markov's Inequality]\label{fact:markovs} For any non-negative random variable $X\geq 0$ and real parameter $a>1$,
    \begin{equation}
        \mathbb{P}[X>a\cdot \mathbb{E}[X]]\leq \frac{1}{a}
    \end{equation}
\end{fact}

\begin{fact}
    [First Moment Method]\label{fact:fmm} For any non-negative random variable $X\geq 0$,
    \begin{equation}
        \mathbb{P}[X>0]\leq \mathbb{E}[X]
    \end{equation}
\end{fact}

\begin{fact}
    [Second Moment Method]\label{fact:smm}
    For any non-negative random variable $X\geq 0$, 
    \begin{equation}
        \mathbb{P}[X>0]\geq \frac{\mathbb{E}[X]^2}{\mathbb{E}[X^2]}
    \end{equation}
\end{fact}

\begin{fact}
    [\cite{tomamichel2017largely}, Lemma 6]\label{fact:tomamichel} Let $Z_1, \cdots Z_{n+k}$ be a collection of binary random variables. Let $S\subset [n+k], |S|=k$ be a uniformly random subset. Then, 

    \begin{equation}
        \prs{}{\sum_{i\in S} Z_i \leq k\cdot \delta \text{ and } \sum_{i\in \Bar{S}}Z_i\geq n(\delta+\nu)} \leq \exp\bigg[-\frac{2\nu^2 nk^2}{(n+k)(k+1)}\bigg]
    \end{equation}
\end{fact}

\subsection{Expander graphs}
\label{sec:expanders}

We rely on a series of basic facts of expander graphs. 

\begin{definition}
    [Spectral Expansion]\label{definition:spectral}
    Let $G = (V, E)$ be an undirected $d$-regular graph, and let $A$ be its adjacency matrix. Let $d=\lambda_1\geq \lambda_2\geq \cdots \geq \lambda_n$ be its eigenvalues. Then $G$ is said to be a $\lambda$-spectral-expander if 
    \begin{equation}
        \max_{i\neq 1} |\lambda_i| \leq  \lambda.
    \end{equation}
\end{definition}

\noindent We often-times will refer to the transition matrix $P = A/d$ of the random walk on $G$, and let $\mu = \lambda/d$. For our reduction we require an explicit deterministic construction. 

\begin{lemma} [\cite{Reingold2000EntropyWT}]\label{lemma:expander_construction}
    There exists explicit constants $m_0, d\in\mathbb{N}$ and a parameter $d>\lambda > 0$, such that there exists a deterministic polynomial-time constructable family $\{G_m\}_{m\geq m_0}$ of $d$-regular $\lambda$-spectral-expander graphs on $m$ vertices for every $m\geq m_0$.
\end{lemma}

\noindent See e.g. \cite[Lemma 2.1]{dinur2007pcp} on how to achieve the condition $\forall m\geq m_0$, from the infinite family of \cite{Reingold2000EntropyWT}.

The constraints on $\lambda$ entail $\mu < 1$. We will require a short lemma on quadratic forms of the transition matrix $P$ of an expander graph. 

\begin{lemma}\label{lemma:expander-quadratic}
    Let $P=A/d$ be the transition matrix of a $d$-regular $(d\cdot \mu)$-spectral-expander graph on $m$ vertices. Then, for all vectors $a, b\in [0, 1]^{\times m}$ and integer $t$:
    \begin{equation}
        a^T P^t b \leq \frac{1}{m} \|a\|_1\cdot  \|b\|_1 +\mu^{t} \big(\|a\|_1+  \|b\|_1\big)
    \end{equation}
\end{lemma}

\begin{proof}
    Let us diagonalize $P = \sum_k \mu_k \gamma_k\gamma_k^T$, where $\mu_1=1$, $|\mu_{k>1}|\leq \mu$, and the eigenvectors $\gamma_k$ form an orthonormal basis.
\begin{align}
    a^T P^{t}b &= \frac{1}{m} \|a\|_1\cdot \|b\|_1 + \sum_k \mu_k^{t}(a\cdot \gamma_k)(b\cdot \gamma_k)  \leq \frac{1}{m} \|a\|_1\cdot \|b\|_1 + \mu^{t} \sum_k |a\cdot \gamma_k||b\cdot \gamma_k|\\
    &\leq \frac{1}{m} \|a\|_1\cdot \|b\|_1 + \mu^{t} \bigg(\sum_k |a\cdot \gamma_k|^2\bigg)^{1/2}\bigg(\sum_k |b\cdot \gamma_k|^2\bigg)^{1/2} \leq \frac{1}{m} \|a\|_1\cdot  \|b\|_1 +\mu^{t}\|a\|_2\cdot  \|b\|_2 
\end{align}

\noindent Where, in sequence, we used $|\mu_k|\leq \mu$, the Cauchy-Schwartz inequality, and the ortho-normality of the $\gamma_k$. Now, 
\begin{equation}
    \|a\|_2\cdot  \|b\|_2 \leq \|a\|_2^2+  \|b\|_2^2\leq \|a\|_1+ \|b\|_1 
\end{equation}
by the AM-GM inequality and since each entry of $a, b$ is $\leq 1$.
\end{proof}

\section{The Main Quantity of Interest: The Violation Number Measurement}
\label{section:viol_number_measurement}

We recollect that we consider families of layered Hamiltonians as in \cref{definition:layered}. Following \cref{def:amplified_ham}, to define the amplified Hamiltonian, we amplify the layers separately:

\begin{equation}
    H = \sum_{\chi\in [g]}w_\chi H_\chi \overset{\mathrm{amplify}}{\longrightarrow} H^{(2t)} \coloneqq \sum_{\chi\in [g]}w_\chi H_\chi^{(2t)}
\end{equation}

\noindent Let the number of projective terms in $H_\chi$ be $m_\chi$. 

For any given color $\chi$ and function $f \in \mathcal{F}_\chi$ (the paths of length $2t$, \cref{definition:paths}), it is useful to associate a measurement $N_f^{\chi}$. $N_f^{\chi}$ counts the number of projections of color $\chi$ that are violated when the projectors specified by $f\in \mathcal{F}_\chi$ are measured, i.e. the projector $\Pi_{f(j)}$ is measured on the $j$th register for all $j \in [2t]$.

\begin{definition}
    [Violation Number Measurement] Fix a color $\chi\in [g]$. Given a state $\rho$ on $2t\cdot n$ qubits and a function $f:[2t]\rightarrow [m_\chi]$:
    \begin{enumerate}
      \item Divide the $2t\cdot n$ qubits into $2t$ blocks of $n$ qubits, and measure each block as follows: On the $j$th block, perform the measurement $\{\mathbb{I}-\Pi^\chi_{f(j)}, \Pi^\chi_{f(j)}\}$ with outcome $b_j\in \{0, 1\}$. 
        \item Output $\sum_j b_j$.
    \end{enumerate}
    \label{def:colored-violation-number-meas}
\end{definition}

\noindent The observable associated with this measurement is denoted
    \begin{equation}\label{equation:n_as_sum}
        N_f^\chi\equiv \sum_{j\in [2t]}(\Pi^\chi_{f(j)})_{\reg j} \otimes \mathbb{I}_{[2t]\setminus j}\;.
    \end{equation}

The notation $(\Pi_{f(j)})_{\reg j}$ means $\Pi_{f(j)}$ acting on register $j$. Intuitively, $N^{\chi}_f$ chooses $2t$ projectors from the Hamiltonian $H_\chi$ according to the given function $f$ and then measures these projectors in parallel to count how many of them are violated on the state $\rho$.

\begin{remark}
    $N^{\chi}_f$ is a Hermitian operator with eigenvalues $\{0, 1, \dots, 2t\}$ whose eigenspaces are the projectors onto the different outcomes of the measurement procedure above. However, as written \cref{equation:n_as_sum} is not the eigendecomposition of $N^{\chi}_f$.
\end{remark}

We also define another collection of measurements, this one indexed by $\chi$ (color), $f$ (function) and also $S \subseteq [2t]$ (subset). Measuring $N_{f, S}^\chi$ will correspond to measuring the projectors specified by $f$ only in the registers whose indices are inside $S$, and counting how many violations result.

\begin{definition}
    [Subset Violation Number] Fix a color $\chi\in [g]$. Given a state $\rho$ on $2t\cdot n$ qubits, a function $f:[2t]\rightarrow [m_\chi]$ and a subset $S \subseteq [2t]$:
    \begin{enumerate}
        \item Divide the $2t\cdot n$ qubits into $2t$ blocks of $n$ qubits, and for $j\in S$ measure $\{\mathbb{I}-\Pi^\chi_{f(j)}, \Pi^\chi_{f(j)}\}$ on the $j$th block with outcome $b_j\in \{0, 1\}$. 
        \item Output $\sum_{j\in S} b_j$.
    \end{enumerate}
    \label{def:colored-subset-violation-number-meas}
\end{definition}

\noindent We denote the associated observable as $N_{f, S}^\chi$ in the natural way. We will be interested in the probability that measuring $N^{\chi}_f$ (or $N^{\chi}_{f, S}$) on some state $\rho$ does \emph{not} yield outcome 0. We will write this probability as
\begin{align*}
\prs{\rho}{N^{\chi}_f > 0} \,.
\end{align*}
The projector onto the 0-eigenspace of $N^{\chi}_f$ is $\bigotimes_{j \in [2t]} (\1 - \Pi^\chi_{f(j)})$, so we have 
\begin{align}
\prs{\rho}{N^{\chi}_f > 0} = 1 - \tr{\bigotimes_{j \in [2t]} (\mathbb{I} - \Pi_{f(j)}^\chi)  \rho} \,. \label{eqn:n0_expanded}
\end{align}

\noindent The following lemmas give the relationship between the energy of $H^{(2t)}_\chi$ and the observables $N^{\chi}_f$ and $ N^{\chi}_{f,S}$.

\begin{lemma} \label{lemma:Nmin}
    Consider $H_\chi^{(2t)}$ and $N_f^\chi$ as defined above. Then, for any state $\rho$:
    \begin{equation}
        \tr{H^{(2t)}_\chi \rho} = \E_{f \in \cF_\chi} \prs{\rho}{N^{\chi}_f > 0}.
    \end{equation}
\end{lemma}
\begin{proof}
    This follows by comparing the definition of $H^{(2t)}_\chi$ and \cref{eqn:n0_expanded}.
\end{proof}

\begin{lemma}\label{lemma:tilde_greater}
    Consider $H_\chi^{(2t)}$, $N_f^\chi$ and $N_{f, S}^\chi$ as defined above. Then, for any function $f$, state $\rho$, and subset $S$,
    \begin{equation}
        \tr{N^{\chi}_{f}\rho} \geq \tr{N^{\chi}_{f,S}\rho},\quad  \prs{\rho}{N^{\chi}_{f} > 0} \geq \prs{\rho}{N^{\chi}_{f,S} > 0}
    \end{equation}
\end{lemma}
\begin{proof}
    The bits $b_j$ that are summed in the measurement procedure for $N^{\chi}_{f,S}$ are a subset of the bits summed in the measurement procedure for $N^{\chi}_f$.
\end{proof}

\section{On the Completeness of the Amplified Hamiltonian}
\label{section:completeness}
In this section, we present an upper bound on the ground state energy of the amplified Hamiltonian. The upper bound says that if the ground state energy of the original Hamiltonian was very low, then the ground state energy of the amplified Hamiltonian remains fairly low.

\begin{proposition} \label{lem:upperbound}
Fix a layered Hamiltonian $H$ as in \cref{definition:layered}, and the collection of function families $\{\mathcal{F}_\chi\}$ from \cref{definition:paths}. For any $t \in \N$, write $H^{(2t)} = \sum_\chi w_\chi H_\chi^{(2t)}$ to denote the derandomized $2t$-fold tensor product amplification of $H$ (\cref{def:amplified_ham}). The lowest eigenvalue of $H^{(2t)}$ is bounded from above by
\begin{align*}
\lambda_{\min}(H^{(2t)}) \leq 2 t \cdot \lambda_{\min}(H) \,.
\end{align*}
\end{proposition}

\begin{proof}
    Let $\rho$ be the ground state of the original Hamiltonian $H = \sum_\chi w_\chi H_\chi$. Let us consider the Hamiltonian terms of each color one at a time. Note that, since we are proving an upper bound on the lowest eigenvalue of the amplified Hamiltonian, it suffices to show that the value is upper bounded on some specific state. Consider the state $\rho^{\otimes 2t}$. From \cref{lemma:Nmin},
    \begin{equation}
        \Tr[H_\chi^{(2t)}\rho^{\otimes 2t}] = \mathbb{E}_{f\in \cF_\chi}\prs{\rho^{\otimes 2t}} {N_{f}^\chi >0} 
    \end{equation}

    \noindent From the first moment method (\cref{fact:fmm}), we have 
    \begin{equation}
        \prs{\rho^{\otimes 2t}} {N_{f}^\chi >0} \leq \tr{N_{f}^\chi \rho^{\otimes 2t}}\Rightarrow \mathbb{E}_{f\in \cF_\chi}\prs{\rho^{\otimes 2t}}{N_{f}^\chi >0} \leq \sum_{i\in [2t]}\tr{(H_\chi)_{\reg i} \rho^{\otimes 2t}} = 2t\cdot \tr{H_\chi\rho}
    \end{equation}
    By linearity, we conclude
    \begin{equation}
        \lambda_{min}(H^{(2t)}) \leq 2t\cdot \sum_\chi w_\chi \tr{H_\chi\rho} =  2t\cdot  \lambda_{min}(H) 
    \end{equation}

\end{proof}

\section{A Technical Tool: The Auxiliary Energy Measurement}
\label{section:aux_energy_meas}

To simplify notation, for $j\in [2t]$, let us denote by $H_{\reg j}$ the application of $H$ on the $j$-th register.

\begin{equation}
    H_{\reg j} = \underbrace{\mathbb{I} \otimes \cdots \otimes \mathbb{I}}_{j-1} \otimes \, H \, \otimes \underbrace{\mathbb{I} \otimes \cdots \otimes \mathbb{I}}_{2t-j}
\end{equation}

In this section, we introduce a measurement which we refer to as the `auxiliary energy measurement'. This measurement is never performed in a real measurement of the amplified Hamiltonian $H^{(2t)}$; it exists purely in the analysis. Loosely speaking, we think of this measurement as a \emph{diagnostic}: we select some random subset of registers on which to perform it, and its outcomes on those `auxiliary' registers will help us determine how to proceed in the analysis of the quantity we care about, namely $N^{\chi}_{f,S}$, which is defined on the remaining (non-auxiliary) registers. The fact that we select the subset of registers randomly will mean that the prover cannot adversarially bias the outcome of the auxiliary measurement in order to mislead us about the non-auxiliary registers.

\begin{definition}
    [The Auxiliary Energy Measurement]\label{definition:aux_meas} Given any state $\rho$ on $n\cdot 2t$ qubits, a subset $S \subseteq [2t]$, and a threshold parameter $\alpha \in \R$:
    \begin{enumerate}
\item For $j\in[2t]$ define $\Pi^{\geq \alpha}_{\reg j}$ as the projection onto the direct sum of all eigenspaces of $H_{\reg j}$ with associated eigenvalue larger than $\alpha$.
\item For all $j \in \overline S$, measure $\{\mathbb{I} - \Pi^{\geq \alpha}_{\reg j}, \Pi^{\geq \alpha}_{\reg j}\}$ and call the outcome $c_j \in \bits$.
\item Output $c = (c_j)_{j \in \overline S}$.
\end{enumerate}
\end{definition}

We will write the probability of receiving an outcome $c$ in this measurement as $\prs{\rho}{c}$, where it will be clear from context that we are referring to the auxiliary energy measurement and what the set $S$ and threshold $\alpha$ are.
We will denote by $\rho_{|S,\alpha,c}$ the post-measurement state after receiving outcome $c$ in the auxiliary measurement with set $S$ and threshold $\alpha$.
Note that in the measurement procedure, we measure all registers \emph{not} in $S$, i.e.~$\rho_{|S,\alpha,c}$ is a post-measurement state whose registers in $S$ have been left untouched.

Our analysis will hinge on a careful case division, where we classify the possible outcomes of the auxiliary energy measurement into two categories, ``high-energy" and ``low-energy", and decide what to do based on whether we got a high-- or low-energy outcome. For this purpose, let $r\in[t]$ be an integer parameter (to be picked soon). For any subset $S \subseteq [2t]$, we define

\begin{equation}
\label{eqn:us_def}
    U_{\Bar{S}}^r\equiv \bigg\{ c\in \{0, 1\}^{\Bar{S}}: |c|\geq 4r\bigg\}  \footnote{We also denote $\overline U_{\Bar{S}}^r\subset \{0, 1\}^{\Bar{S}}$ as the complement set. Whenever implicit, the superscript $r$ is omitted for conciseness. }
\end{equation}

The following lemma captures the intuition that our auxiliary measurement is a good diagnostic because we pick the subset on which to perform it randomly. Suppose we were to extend the auxiliary energy measurement to \textbf{all} the registers, i.e., suppose we were to perform it in $S$ as well as in $\bar S$. We expect that, since $S$ is picked randomly, it is likely that, when we condition on receiving a high energy outcome in $\bar S$, we will also receive a high energy outcome in $S$ with decent probability.

\begin{lemma}\label{lemma:de-finetti}
For each $j\in [2t]$, let the binary random variable $C_j$ denote the outcome of the measurement $\{\1 - \Pi^{\geq \alpha}_{\reg j}, \Pi^{\geq \alpha}_{\reg j}\}$ on any $2t\times n$ qubit state $\rho$. Then, for every $r\in [t]$,
    \begin{equation}
        \mathbb{E}_{|S|=t}\Pr_\rho \Big[\sum_{j \in S} C_j < 2r \; \text{ and }\; \sum_{i \in \overline S} C_i \geq 4r\Big]   \leq e^{-\frac{2 r^2}{t}}
    \end{equation}
\end{lemma}

\begin{proof}
    Here we remark that the $\{C_j\}$ correspond to the outcomes of a sequence of commuting measurements, so their joint distribution is well-defined. The proof then follows immediately from \cref{fact:tomamichel}.
\end{proof}

As we later discuss, this is the central ``de Finetti" statement we require in our approach.

\begin{remark}
    So long as one picks the same set $S$ in both, the subset violation number measurement and the auxiliary energy measurement commute. We have therefore the following decomposition for any state $\rho$, color $\chi$, threshold $\alpha$, subset $S$ and function $f$:
\begin{align}
\prs{\rho}{ N^{\chi}_{f,S} > 0} 
&= \sum_{c\in \bits^{\overline{S}}} \prs{\rho}{c} \ \prs{\rho_{|S,\alpha,c}}{ N^{\chi}_{f,S} > 0} \,. \label{eqn:tildeN_expansion}
\end{align}
\end{remark}

    Our final lower bound on the ground energy of the amplified Hamiltonian, introduced in \cref{sec:soundness}, is based solely on the behavior of $N^{\chi}_{f,S}$, which we manage to control using the auxiliary measurements. We are `sacrificing' the auxiliary registers, in the sense that we do not count violations on them except as a diagnostic tool. We can do this because the outcome of $N^{\chi}_{f,S}$ always lower bounds that of $N^{\chi}_{f}$ (\cref{lemma:tilde_greater}).

\section{On the Soundness of the Amplified Hamiltonian}
\label{sec:soundness}
In this section, we prove a lower bound on the ground energy of the amplified Hamiltonian: we show that one cycle of amplification boosts the ground energy of the original Hamiltonian by at least a constant multiplicative factor which depends on $t$, i.e.~the number of `copies' of the original Hamiltonian. Together with \Cref{lem:upperbound}, this shows that our amplification procedure amplifies the promise gap of any Hamiltonian with sufficiently good completeness.

The key challenge is that the ground state of the amplified Hamiltonian may not be a product state, and therefore may be entangled in a way that causes its energy to be lower than what a purely classical analysis of a similar gap amplification procedure, such as that of \cite{dinur2007pcp}, would lead us to expect. Nevertheless, by combining ingredients from \cite{dinur2007pcp}'s proof of classical gap amplification with techniques inspired by the proof of the exponential de Finetti theorem~\cite{renner2007symmetry,renner2008security,vidick2016simple}, we still are able to show amplification of the ground state energy:

\begin{theorem}[Simplified statement of \Cref{cor:lowerbound_simpler}] \label{theorem:lowerbound_simpler}
Fix a layered Hamiltonian $H$ as in \cref{definition:layered}, and the collection of function families $\{\mathcal{F}_\chi\}$ from \cref{definition:paths}. For any $t \in \N$, write $H^{(2t)} = \sum_\chi w_\chi H_\chi^{(2t)}$ to denote the derandomized $2t$-fold tensor product amplification of $H$ (\cref{def:amplified_ham}). The lowest eigenvalue of $H^{(2t)}$ is bounded from below by
\begin{align*}
\lambda_{\min}(H^{(2t)}) \geq \min\Big\{ \Theta\Big(\frac{\log t}{t}\Big) \:\:, \:\: \Theta\Big(\lambda_{\min}(H) \cdot \frac{t^{1/2}}{(\log t)^{1/2}}\Big) \Big\},
\end{align*}
where the big-O hides only constants that do not depend on $t$.
\end{theorem}

\subsection{An overview of the analysis}

The following lemma, \cref{lem:lowerbound}, clarifies the structure of our proof of \cref{theorem:lowerbound_simpler}. Our proof relies on the use of the auxiliary measurement introduced in \cref{section:aux_energy_meas} as a `diagnostic': we prove two independent lower bounds, \cref{eqn:high-energy-bound} and \cref{eq:direct-bound-case}--\cref{eqn:low_energy_bound}, on the amplified energy $\lambda_{\min}(H^{(2t)})$, and then, depending on the outcome of the auxiliary measurement, we decide which one to use. In particular, if the auxiliary energy measurement gives us a high-energy outcome, then we use \cref{eqn:high-energy-bound}, whose proof is a formalisation of the intuition that, if the diagnostic gave us a high-energy outcome, then we expect a high-energy outcome on the registers we care about (the `primary registers') as well: this allows us to prove a direct lower bound on the ground energy. 

On the other hand, if the auxiliary measurement gives us a low-energy outcome, then we need to work a little harder. The proof of the lower bound in the `low-energy case' (expressed in \cref{eq:direct-bound-case}--\cref{eqn:low_energy_bound}) is more closely analogous to Dinur's analysis of classical gap amplification, and it contains most of our technical ideas. One way to intuitively interpret our strategy is that conditioning on a low-energy outcome helps us curtail (or upper bound) the \emph{variance} of $N_f^{\chi}$, and it is easier to argue that the variance is small when $N_f^{\chi}$ is small overall (which is the case in the low-energy branch of the dichotomy).

\begin{proposition} \label{lem:lowerbound}
Let $H$ be a normalised sum of projectors on $n$ qubits as in \cref{equation:H_definition}, and fix any choice of parameters $t\in \mathbb{N}$, $r\in [t]$, $\alpha>0$, and the collection of expander walks $\{\mathcal{F}_\chi\}_{\chi\in [g]}$. Let $\rho$ be any ground state of $H^{(2t)}$, and define the probability of a high energy outcome:
\begin{align}
\Delta \coloneqq \E_{S\subseteq[2t], |S|=t} \sum_{c\in U_{\overline S}} \prs{\rho}{c}
\end{align}
\noindent with $U_{\bar{S}}$ as in \cref{eqn:us_def}. Let $C_\mu$ be the constant in \Cref{lemma:second_moment_Xs}, and let $\omega_{\min}$ be as in \Cref{eq:min-weight-def}.

Then, the lowest eigenvalue of the derandomised $t$-fold tensor product amplification $H^{(2t)}$ is bounded from below by
\begin{align}
\lambda_{\min}(H^{(2t)}) \geq \max\Big\{  \frac{2 \alpha r}{t}\Big(\Delta - e^{-2 r^2/t}\Big) , \;
& \frac{t (1-\Delta)}{1 + C_\mu \: + \: \omega_{\min} \cdot (1 + 8r + \alpha t + 2t \cdot e^{-8r^2/t})} \cdot \lambda_{\min}(H) \Big\}
\label{eqn:lower-bound}
\end{align}
\end{proposition}

In \cref{section:parameters}, we show how to carefully instantiate the parameters $\alpha, r$ and $\Delta$ to prove \cref{theorem:lowerbound_simpler}. For now, we discuss the proof of \cref{lem:lowerbound}.

\begin{proof}[Proof of \cref{lem:lowerbound}] Let $\rho$ be a ground state of $H^{(2t)}$. From \cref{lemma:Nmin} and \cref{lemma:tilde_greater}, we know that 
\begin{equation}
\label{equation:amp_gse_def}
\lambda_{\min}(H^{(2t)}) =  \sum_\chi w_\chi\cdot \tr{H^{(2t)}_\chi \rho} \geq \sum_\chi w_\chi\cdot \bigg( \E_{S \subset [2t], |S| = t}\E_{f \in \cF_\chi} \prs{\rho}{N^{\chi}_{f,S} > 0} \bigg)
\end{equation}

To proceed, let us perform the auxiliary energy measurement on $\Bar{S}$ and the expansion described in \cref{eqn:tildeN_expansion}, for some fixed energy threshold $\alpha$ to be determined:
\begin{align}
~\eqref{equation:amp_gse_def}\geq \sum_\chi w_\chi \cdot \E_{|S|=t} \E_{f \in \cF} \sum_{c\in \bits^{\overline{S}}} \prs{\rho}{c} \ \prs{\rho_{|S,\alpha,c}}{N^{\chi}_{f,S} > 0}  \label{eqn:expanded_bound}
\end{align}

Our analysis will be based on a careful case division, in which we split the sum over $c$ in two: one sum over ``high-energy outcomes'' and one sum over ``low-energy outcomes'' of the auxiliary energy measurement. For this purpose, let $r \in [t]$ be an integer parameter to be chosen later. For any subset $S \subseteq [2t]$, we define 
\begin{align}
U_{\overline S} = \{ c \in \bits^{\overline{S}} : |c|\geq 4r\}, \text{ as in \cref{eqn:us_def}} \,. 
\end{align}
Then, 
\begin{align}
\text{~\eqref{eqn:expanded_bound}} = \sum_\chi w_\chi\cdot \bigg( \E_{|S|=t} \E_{f \in \cF} \sum_{c\in U_{\overline S}} \prs{\rho}{c} \ \prs{\rho_{|S,\alpha,c}}{ N^{\chi}_{f,S} > 0} + \E_{|S|=t} \E_{f \in \cF} \sum_{c\in \overline U_{\overline S}} \prs{\rho}{c} \ \prs{\rho_{|S,\alpha,c}}{N^{\chi}_{f,S} > 0}\bigg) \,. \label{eqn:split_sum}
\end{align}
Note that both terms in the sum are always non-negative. We proceed by proving two different lower bounds, one on the first term in \Cref{eqn:split_sum} and one on the second term in \Cref{eqn:split_sum}; one of these bounds will be useful when $\Delta$ is large, i.e.~when a `high-energy' outcome is likely, and the other will be useful when $\Delta$ is small, when a `high-energy' outcome is unlikely. We defer the actual case split analysis to \Cref{section:parameters}.

\paragraph{Case 1: high-energy outcome is likely.}
The following bound will be useful when we expect $\Delta$ to be large (say, bigger than some fixed constant, such as $\frac{1}{2}$). For this case we will make use of the following general bound, which we will prove below in \cref{lem:high_energy_likely}:
\begin{align}
\sum_\chi w_\chi\cdot \bigg(\E_{S\subseteq[2t], |S|=t} \E_{f \in \cF} \sum_{c\in U_{\overline S}} \prs{\rho}{c}   \prs{\rho_{|S,\alpha,c}}{N^{\chi}_{f,S} >0}\bigg) \, &\geq\,  \frac{2 \alpha r}{t}\Big(\E_{S} \sum_{c\in U_{\overline S}} \prs{\rho}{c}-e^{-2 r^2/t}\Big) \\&=\,  \frac{2 \alpha r}{t}\Big(\Delta - e^{-2 r^2/t}\Big)\,. \label{eqn:high-energy-bound} 
\end{align}

\paragraph{Case 2: high-energy outcome is unlikely.}
The following bound will be useful when we expect $\Delta$ to be small. For this case we will make use of the following general bound, which we will prove in \cref{lem:high_energy_unlikely}: either it is the case that we get the direct bound on the amplified energy
\begin{align}
    \tr{H^{(2t)} \rho} &\geq t \lambda_{\min}(H) \cdot \E_{|S|=t} \prs{\rho}{\overline{U}_{\overline{S}}} \\
    &= \lambda_{\min}(H) \cdot t(1-\Delta),
    \label{eq:direct-bound-case}
\end{align}
or it is the case that
\begin{align*}
&\sum_\chi w_\chi\cdot \bigg(\E_{S\subseteq[2t], |S|=t} \E_{f \in \cF_\chi} \sum_{c\in \overline U_{\overline S}} \prs{\rho}{c} \ \prs{\rho_{|S,\alpha,c}}{N^{\chi}_{f,S} >0} \bigg) \\
&\geq\,\lambda_{\min}(H)\times \frac{ t \E_{|S| = t} \prs{\rho}{\overline{U}_{\overline{S}}}}{1 + C_\mu \: + \: \omega_{\min} \cdot (1 + 8r + \alpha t + 2t \cdot e^{-8r^2/t})} \\
&=\, \lambda_{\min}(H)\times \frac{t (1-\Delta)}{1 + C_\mu \: + \: \omega_{\min} \cdot (1 + 8r + \alpha t + 2t \cdot e^{-8r^2/t})} \numberthis \label{eqn:low_energy_bound}
\end{align*}

Substituting our two bounds \Cref{eqn:high-energy-bound} and \Cref{eqn:low_energy_bound} into \Cref{eqn:split_sum}, and keeping in mind the special case expressed in \Cref{eq:direct-bound-case}, we find the following lower bound: either
\begin{align}
\lambda_{\min}(H^{(2t)}) &\geq \lambda_{\min}(H) \cdot t(1-\Delta),
\label{eq:special-case}
\end{align}
or
\begin{align}
\lambda_{\min}(H^{(2t)}) \geq  \frac{2 \alpha r}{t}\Big(\Delta - e^{-2 r^2/t}\Big) + \;
& \frac{t (1-\Delta)}{1 + C_\mu \: + \: \omega_{\min} \cdot (1 + 8r + \alpha t + 2t \cdot e^{-8r^2/t})} \cdot \lambda_{\min}(H).
\label{eq:normal-case}
\end{align}

In addition, \cref{eqn:high-energy-bound} is always true and constitutes another independent lower bound on $\lambda_{\min}(H^{(2t)})$. Hence we have the following composite lower bound on $\lambda_{\min}(H^{(2t)})$:
\begin{equation}
    \lambda_{\min}(H^{(2t)}) \geq \max\Big\{\eqref{eqn:high-energy-bound}, \: \min\{ \eqref{eq:special-case}, \eqref{eq:normal-case} \} \Big\}
    \label{eq:lower-bound-composite}
\end{equation}
\Cref{eqn:lower-bound} is a lower bound on this composite bound \Cref{eq:lower-bound-composite} because the second term in the max in \Cref{eqn:lower-bound} is always smaller than the second term in the max in \Cref{eq:lower-bound-composite}.
\end{proof}

In the following two subsections, we prove the two bounds that we use in the proof of \cref{lem:lowerbound}: \cref{lem:high_energy_likely} (which lower bounds the first term in \cref{eqn:split_sum}) and \cref{lem:high_energy_unlikely} (which lower bounds the second term in \cref{eqn:split_sum}). Intuitively speaking, \cref{lem:high_energy_likely} is the lower bound that we will apply when we find that the auxiliary measurement gives us a `high-energy' outcome, and \cref{lem:high_energy_unlikely} is the lower bound we will apply when the auxiliary measurement gives us a `low-energy' outcome.

\subsection{The high-energy case}

In this subsection we prove \cref{lem:high_energy_likely}, which gives a lower bound on the first term in \cref{eqn:split_sum}, namely, the term that conditions on getting a high-energy outcome from the auxiliary measurement introduced in \cref{section:aux_energy_meas}. The key technical ingredient in this proof is \cref{lemma:de-finetti}, which implies that it is unlikely for the `primary' registers to be low-energy if the `auxiliary' qubits were high-energy.

\begin{lemma}[The High Energy Case]\label{lem:high_energy_likely}
For any choice of parameters $\alpha > 0$ and $r \in [t]$ and any state $\rho$ on $2t\cdot n$ qubits it holds that
\begin{align*}
\sum_{\chi\in [g]}w_\chi\cdot \bigg(\E_{S\subseteq[2t], |S|=t} \E_{f \in \cF_\chi} \sum_{c\in U_{\overline S}} \prs{\rho}{c} \ \prs{\rho_{|S,\alpha,c}}{N^{\chi}_{f,S} >0} \, \bigg)\geq\,  \frac{2 \alpha r}{t}\Big(\E_{S} \sum_{c\in U_{\overline S}} \prs{\rho}{c}-e^{-\frac{2 r^2}{t}}\Big)\,.
\end{align*}
Here, $U_{\overline S}$ is defined as in \cref{eqn:us_def} and depends implicitly on $\alpha$ and $r$.
\end{lemma}

\begin{proof}

Using \cref{eqn:n0_expanded}, for every $\rho_{|S,\alpha,c}$, we have the following naive lower bound:
\begin{align}
    \prs{\rho_{|S,\alpha,c}}{N^{\chi}_{f,S} >0} 
&= 1 - \tr{\bigotimes_{j \in S} (\1 - \Pi^\chi_{f(j)}) \rho_{|S,\alpha,c}} \geq 1 - \min_{j \in S} \tr{ (\1 - \Pi^\chi_{f(j)})_{\reg j} \rho_{|S,\alpha,c}} \\
&\geq \max_{j \in S} \tr{(\Pi^\chi_{f(j)})_{\reg j} \rho_{|S,\alpha,c}} \geq \frac{1}{t} \sum_{j\in S}\tr{(\Pi^\chi_{f(j)})_{\reg j} \rho_{|S,\alpha,c}},
\end{align}
\noindent where we used $\bigotimes_{i \in S} (\1 - \Pi_{f(i)}) \leq (\1 - \Pi_{f(j)})_{\reg j}$ for any $j \in S$. In expectation over $f\in \cF_\chi$, 
\begin{equation}
    \E_{f \in \cF_\chi} \prs{\rho_{|S,\alpha,c}}{N^{\chi}_{f,S} >0}  \geq \frac{1}{t} \sum_{j\in S}\tr{(H_\chi)_{\reg j} \rho_{|S,\alpha,c}},
\end{equation}
since the random variable $f(j)$ is uniformly distributed. Now, we can group the terms of the different colors $\chi\in [g]$:
\begin{equation}
    \sum_{\chi\in [g]} w_\chi \E_{f \in \cF_\chi} \prs{\rho_{|S,\alpha,c}}{N^{\chi}_{f,S} >0} \geq \frac{1}{t}  \sum_{j\in S}\tr{(H)_{\reg j} \rho_{|S,\alpha,c}}  \geq \frac{\alpha}{t}  \sum_{j\in S}\tr{\Pi^{\geq \alpha}_{\reg j} \rho_{|S,\alpha,c}} 
\end{equation}

\noindent Since $\Pi^{\geq \alpha}_{\reg j} \leq \frac{1}{\alpha} H$. The above expression considers the probability of receiving an energy above $\alpha$ on a register in $S$ conditioned on having received an energy above $\alpha$ on at least $4r$ registers in $\overline{S}$. Following the discussion in \cref{section:aux_energy_meas}, define $C_j$ as the binary random variable corresponding to the measurement $\{\Pi^{\geq \alpha}_{\reg j}, \1 - \Pi^{\geq \alpha}_{\reg j}\}$.
\begin{align*}
\sum_{c\in U_{\overline S}} \prs{\rho}{c} \sum_{j \in S} \tr{ \Pi^{\geq \alpha}_{\reg j} \, \rho_{|S,\alpha,c} }
= \sum_{j \in S} \tr{ \Pi^{\geq \alpha}_{\reg j} \, \left( \sum_{c\in U_{\overline S}} \prs{\rho}{c} \rho_{|S,\alpha,c} \right) }
= \sum_{j \in S} \prs{\rho}{C_j = 1 \;\wedge\; \sum_{i \in \overline S} C_i \geq 4r} \,.
\end{align*}

This is because $\sum_{c\in U_{\overline S}} \prs{\rho}{c} \rho_{|S,\alpha,c}$ is the \emph{subnormalised} post-measurement state for at least $4r$ of the measurements associated with $\{C_i\}_{i \in S}$ yielding 1.
We can further bound 
\begin{align*}
\sum_{j \in S} \Pr\Big[C_j = 1 \;\wedge\; \sum_{i \in \overline S} C_i \geq 4r\Big]
&= \E \Big[\sum_{j \in S} C_j \;\Big|\; \sum_{i \in \overline S} C_i \geq 4r\Big] \cdot \Pr\Big[ \sum_{i \in \overline S} C_i \geq 4r \Big] \\
&\geq 2r \Pr \Big[\sum_{j \in S} C_j \geq 2r \;\Big|\; \sum_{i \in \overline S} C_i \geq 4r\Big] \cdot \Pr\Big[ \sum_{i \in \overline S} C_i \geq 4r \Big] \tag{Markov's, \cref{fact:markovs}} \\
&= 2r \left( 1 - \Pr \Big[\sum_{j \in S} C_j < 2r \;\Big|\; \sum_{i \in \overline S} C_i \geq 4r\Big]  \right) \cdot \Pr\Big[ \sum_{i \in \overline S} C_i \geq 4r \Big] \\
&= 2r \left( \Pr\Big[ \sum_{i \in \overline S} C_i \geq 4r \Big]  - \Pr \Big[\sum_{j \in S} C_j < 2r \;\wedge\; \sum_{i \in \overline S} C_i \geq 4r\Big] \right) \,.
\end{align*}
Combining all of these steps and inserting $\Pr\Big[ \sum_{i \in \overline S} C_i \geq 4r \Big] = \sum_{c\in U_{\overline S}} \prs{\rho}{c}$, we have that 
\begin{align*}
\E_{S\subseteq[2t], |S|=t} \sum_{c\in U_{\overline S}} \prs{\rho}{c} \ \prs{\rho_{|S,\alpha,c}}{ N^{(2t)}_{f,S} >0} \geq \frac{2 \alpha r}{t} \left( \E_{S} \sum_{c\in U_{\overline S}} \prs{\rho}{c} - \E_S \Pr \Big[\sum_{j \in S} C_j < 2r \;\wedge\; \sum_{i \in \overline S} C_i \geq 4r\Big] \right)
\end{align*}
Finally, we can apply \cref{lemma:de-finetti} to conclude the proof.
\end{proof}

\subsection{The low-energy case}
\label{sec:low-energy-analysis}

In this subsection we prove \cref{lem:high_energy_unlikely}, which controls the contribution to the energy (in \cref{eqn:split_sum}) when conditioned on receiving a low-energy outcome from the auxiliary measurement introduced in \cref{section:aux_energy_meas}. This section (together with \cref{section:deferred_claims}, which contains the proofs of the most technically difficult lemmas that we use in this section) constitutes the centerpiece of our analysis. 

Our broad strategy for proving \cref{lem:high_energy_unlikely} is similar to Dinur's strategy in \cite[Section 6]{dinur2007pcp}: the lower bound claimed in \cref{lem:high_energy_unlikely} follows from a lower bound on the mean and an upper bound on the variance of the `violation number' (which for us is a random variable associated with the outcome of a measurement) and the Second Moment Method (\cref{fact:smm}). Controlling the variance, however, is highly non-trivial and requires several modications to the classical strategy. Our approach can be broken into two main steps, which are encapsulated in two lemmas: \cref{lemma:second_moment_Xs} and \cref{lem:2ndmoment}. Both have proofs which are deferred to \cref{section:deferred_claims}; in this section we simply present their statements, explain their intuition, and show how to combine them to control the `low-energy' component of $\lambda_{\min}(H^{(2t)})$, namely, the second term of \cref{eqn:split_sum}. We continue this overview below \cref{definition:weighted_violation}, after we have defined some notation.

We firstly state the main lemma that we prove in this section.

\begin{lemma}[The Low Energy Case]\label{lem:high_energy_unlikely}
Let $\rho$ be a ground state of the amplified Hamiltonian $H^{(2t)}$. Let $C_\mu$ be the constant in \Cref{lemma:second_moment_Xs}, and let $\omega_{\min}$ be as in \Cref{eq:min-weight-def}.
For any choice of parameters $\alpha > 0$ and $r \in [t]$ it holds that \emph{either}
\begin{align}
    \tr{H^{(2t)} \rho} \geq t \lambda_{\min}(H) \cdot \E_{|S|=t} \prs{\rho}{\overline{U}_{\overline{S}}}
\label{eq:lowenergy_lowerbound_trivial}
\end{align}
\emph{or}
\begin{align}
\sum_\chi w_\chi \cdot \E_{|S|=t} \E_{f \in \cF_\chi} \sum_{c\in \overline U_{\overline S}} \prs{\rho}{c} \ \prs{\rho_{|S,\alpha,c}}{ N^{\chi}_{f,S} >0} \,\geq\, \frac{t\cdot  \E_{|S| = t} \prs{\rho}{\overline{U}_{\overline{S}}}}{1 + C_\mu \: + \: \omega_{\min} \cdot (1 + 8r + \alpha t + 2t \cdot e^{-8r^2/t})} \times \lambda_{\min}(H) \,. \label{eqn:lowenergy_lowerbound}
\end{align}
Here, $U_{\overline S}$ is defined as in \cref{eqn:us_def} and depends implicitly on $\alpha$ and $r$, and $\prs{\rho}{\overline U_{\overline S}} = \sum_{c \in \overline U_{\overline S}}\prs{\rho}{c}$.
\end{lemma}

It will be helpful to define and analyse the following random variable $X_S$, which computes a weighted average of the ``Subset Violation Number" (see \cref{def:colored-subset-violation-number-meas}) over a random choice of the color $\chi$, with the weights determined by the weights $w_\chi$ of the layers in the Hamiltonian $H$.

\begin{definition}
    [Weighted Violation Number]\label{definition:weighted_violation} Let $\sigma$ be an arbitrary state on $2t$ $n$-qubit registers. Consider the following random variable $X_{S}$, defined by measuring $\sigma$ as follows:
\begin{enumerate}
    \item Pick a color $\chi$ by flipping a $g$-sided die with weights $(w_1, \cdots, w_g)$.
    \item Pick a function $f\in\mathcal{F}_{\chi}$.
    \item Measure the observable $N_{f, S}^{\chi}$ (\cref{def:colored-subset-violation-number-meas}).
\end{enumerate}
\end{definition}

\noindent Our analysis proceeds in three steps. 

\begin{enumerate}
    \item[1.] We begin by proving \cref{lemma:expectation_XS}, which is a simple analysis of the expectation value of $X_S$.
    \item[2.] Subsequently, \cref{lemma:second_moment_Xs} attempts to analyze the variance of $X_S$. Roughly speaking, \cref{lemma:second_moment_Xs} claims that the variance of $X_S$, when the new (amplified) clause $f$ is sampled using an expander random walk, is approximately the same as if $f$ were \textit{pairwise independent} (up to some loss dependent on the expander graph and the number of non-commuting layers).
    \item[3.] \cref{lem:2ndmoment} (essentially) bounds the variance of $X_S$ if $f$ were in fact drawn from a pairwise independent function family. The proof of \cref{lem:2ndmoment} is the biggest departure that our quantum analysis makes from the analogous classical analysis,\footnote{A claim analogous to \cref{lem:2ndmoment} is easy to prove if the ground state is a product state, which it would be if it were a classical string. The presence of entanglement allows the variance of non-local measurements to increase via constructive interference.} and is where we crucially leverage the de Finetti framework developed in the previous sections. 
\end{enumerate}

Put together, \cref{lemma:second_moment_Xs} and \cref{lem:2ndmoment} give us control over the variance of $X_S$; which together with \cref{lemma:expectation_XS} allow us to apply the second moment method in order to lower bound the probability that $X_S$ is greater than 0. 

\begin{lemma}[The first moment of $X_S$]\label{lemma:expectation_XS}
    For any state $\sigma$ on $2t$ $n$-qubit registers, the expectation of $X_S$ satisfies 
    \begin{equation}
       \mathbb{E}[X_S] =\sum_{j\in S}\Tr[H_{\reg j}\sigma]
    \end{equation} 
\end{lemma}

\begin{proof}
    The definition entails
    \begin{equation}
        \mathbb{E}[X_S] = \sum_{\chi\in [g]} w_\chi \cdot \mathbb{E}_{f\in \mathcal{F}_\chi} \Tr[N_{f, S}^\chi\sigma]
    \end{equation}
    The expectation $\mathbb{E}[X_S]$ follows from linearity of expectation, and the fact that over random choice of $f\in \mathcal{F}_\chi$, $f(j)$ is uniformly random over $[m_\chi]$: 
\begin{equation}
    \mathbb{E}_f \Tr[N_{f, S}^\chi\sigma] = \sum_{j\in S}\Tr[H_{\reg j}^\chi\sigma].
\end{equation}
\end{proof}

The next two lemmas, \cref{lemma:second_moment_Xs} and \cref{lem:2ndmoment}, capture our analysis of the variance of $X_S$. An overview of how to interpret their statements is given at the start of \cref{sec:low-energy-analysis}.

\begin{lemma}
    [The second moment of $X_S$]\label{lemma:second_moment_Xs} There exists a constant $C_\mu\leq 2/(1-\mu)$ dependent just on the collection of expander graphs, such that for every state $\sigma$ on $2t$ $n$-qubit registers, 
    
    \begin{align}
        \E[X_S^2]\leq (1+C_\mu)\sum_{j\in S}\Tr[H_{\reg j}\sigma] + \omega_{\min}\cdot  \sum_{i\neq j\in S}\Tr[H_{\reg i}\otimes H_{\reg j}\sigma]
        \label{eq:second-moment-Xs}
    \end{align}
\end{lemma}

For clarity, we defer the proof of \cref{lemma:second_moment_Xs} to the following \cref{section:deferred_claims}. 

The following lemma (essentially) quantifies the variance of $X_S$ if $f$ were drawn from a pairwise independent function family, modulo a technical condition (\cref{eqn:trivial_single_copy_bound}) on how the energy of $\rho$ balances across the registers. 

\begin{lemma} \label{lem:2ndmoment}
Suppose that $\rho$ is a state such that for all $i$, 
\begin{align}
\tr{H_{\reg i} \rho} \leq \E_{|S|=t} \sum_{j \in S} \sum_{c\in \overline U_{\overline S}} \prs{\rho}{c} \ \tr{H_{\reg j} \rho_{|S,\alpha,c}} \,. \label{eqn:trivial_single_copy_bound}
\end{align}
Then for any choice of parameters $\alpha > 0$ and $r \in [t]$ it holds that
\begin{align*}
\E_{|S|=t} \sum_{c\in \overline U_{\overline S}} \prs{\rho}{c} \sum_{\substack{i,j \in S \\ i \neq j}} \tr{H_{\reg i} \cdot H_{\reg j} \rho_{|S, \alpha, c} }
\leq (8r + \alpha t + 2t \cdot e^{-\frac{8r^2}{t}}) \cdot \E_{|S|=t} \sum_{j \in S} \sum_{c\in \overline U_{\overline S}} \prs{\rho}{c} \tr{H_{\reg j} \rho_{|S, \alpha, c} } \,.
\end{align*}
Here, $U_{\overline S}$ is defined as in \cref{eqn:us_def} and depends implicitly on $\alpha$ and $r$.
\end{lemma}

Again, for clarity of structure, we defer the proof of \cref{lem:2ndmoment} to \cref{section:deferred_claims}. Instead, we show how to combine Lemmas \ref{lemma:second_moment_Xs} and \ref{lem:2ndmoment} to conclude the proof of \cref{lem:high_energy_unlikely}, which lower bounds the contributions that the `low energy terms' make to the ground energy of the amplified Hamiltonian.

\begin{proof}

[of \cref{lem:high_energy_unlikely}] Suppose that there exists an $i$ for which 
\begin{align*}
\tr{H_{\reg i} \rho} \geq \E_{|S|=t} \sum_{j \in S} \sum_{c\in \overline U_{\overline S}} \prs{\rho}{c} \ \tr{H_{\reg j} \rho_{|S,\alpha,c}} \,.
\end{align*} 
Then using that $\tr{H_{\reg j} \rho_{|S,\alpha,c}} \geq \lambda_{\min}(H)$, we get that 
\begin{align*}
 \tr{H^{(2t)} \rho} \geq \tr{H_{\reg i} \rho} \geq t \lambda_{\min}(H) \cdot \E_{|S|=t} \sum_{c\in \overline U_{\overline S}} \prs{\rho}{c} = t \lambda_{\min}(H) \cdot \E_{|S|=t} \prs{\rho}{\overline{U}_{\overline{S}}} \,,
\end{align*}

\noindent which implies \cref{eq:lowenergy_lowerbound_trivial}. Thus, for the remainder of the proof we will consider the case that for all $i$, 
\begin{align}
\tr{H_{\reg i} \rho} \leq \E_{|S|=t} \sum_{j \in S} \sum_{c\in \overline U_{\overline S}} \prs{\rho}{c} \ \tr{H_{\reg j} \rho_{|S,\alpha,c}} \,. \label{eqn:trivial_single_copy_bound2}
\end{align}

\noindent Let us now define a random variable $X$ with range $\{0, \cdots, t\}$:

\begin{definition}[definition of $X$]
\label{def:random-variable-X}
\quad
\begin{enumerate}
\item Sample $S \subset [2t]$ of size $|S| = t$.
\item Perform the auxiliary energy measurement on state $\rho$ on registers in $\overline S$, resulting in $c\in \bits^{\overline S}$.
\item If $c \in U_{\overline S}$, output 0. Else, output the random variable $X_S$, defined in \cref{definition:weighted_violation}.
\end{enumerate}
\end{definition}

This random variable simply ``wraps up'' the steps of sampling $S$ and $f$, the auxiliary energy measurement, and the measurement of $N^{\chi}_{f,S}$ into one random variable.
By construction, the desired quantity in the lemma statement can be written in terms of $X$:
\begin{align*}
\sum_\chi w_\chi \E_{|S|=t} \E_{f \in \cF_\chi} \sum_{c\in \overline U_{\overline S}} \prs{\rho}{c} \ \prs{\rho_{|S,\alpha,c}}{N^{\chi}_{f,S} >0} = \pr{X > 0} \,.
\end{align*}

Our goal is to apply the second moment method (\cref{fact:smm}). We note
\begin{align}
\E[X] = \E_{|S|=t} \sum_{c\in \overline U_{\overline S}} \prs{\rho}{c} \E_{\rho_{|S,\alpha,c}}[X_S] =  \E_{|S|=t} \sum_{j \in S} \sum_{c\in \overline U_{\overline S}} \prs{\rho}{c} \ \tr{H_{\reg j} \rho_{|S,\alpha,c}} \,. \label{eqn:numerator_expansion}
\end{align}
\noindent The first equality is by definition of $X$ and the second equality follows from \cref{lemma:expectation_XS}. In turn, one can upper-bound the second moment via \cref{lemma:second_moment_Xs} and \cref{lem:2ndmoment}, under the assumption above:\\
\begin{align*}
\E[X^2] &= \E_{|S|=t} \sum_{c\in \overline U_{\overline S}} \prs{\rho}{c} \E_{\rho_{|S,\alpha,c}}[X_S^2]= \\
&= \E_{|S|=t} \sum_{c\in \overline U_{\overline S}} \prs{\rho}{c} \ \bigg( (1+C_\mu)\sum_{j \in S} \tr{H_{\reg j} \rho_{|S,\alpha,c}} + \omega_{\min}\cdot \sum_{\substack{i,j \in S \\ i \neq j}} \tr{H_{\reg i} \cdot H_{\reg j} \rho_{|S,\alpha,c}} \bigg) \\
&\leq \big(1 + C_\mu \: + \: \omega_{\min} \cdot (1 + 8r + \alpha t + 2t \cdot e^{-8r^2/t})\big) \Big( \E_{|S|=t} \sum_{c\in \overline U_{\overline S}} \prs{\rho}{c} \sum_{j \in S} \tr{H_{\reg j} \rho_{|S,\alpha,c}}\Big) \\
&= \big(1 + C_\mu \: + \: \omega_{\min} \cdot (1 + 8r + \alpha t + 2t \cdot e^{-8r^2/t})\big) \E[X] \,.
\end{align*}
For the second line we used \cref{lemma:second_moment_Xs}, for the third line we used \cref{lem:2ndmoment} (which is applicable since we assumed \cref{eqn:trivial_single_copy_bound2}), and the last line follows by \cref{eqn:numerator_expansion}. We conclude
\begin{align*}
\pr{X > 0} \geq \frac{\E[X]^2}{\E[X^2]} &\geq \frac{\E[X] }{1 + C_\mu \: + \: \omega_{\min} \cdot (1 + 8r + \alpha t + 2t \cdot e^{-8r^2/t})} \\  &\geq \frac{\lambda_{\min}(H)\cdot  t \E_{|S| = t} \prs{\rho}{\overline{U}_{\overline{S}}}}{1 + C_\mu \: + \: \omega_{\min} \cdot (1 + 8r + \alpha t + 2t \cdot e^{-8r^2/t})} \,.
\end{align*}
\end{proof}

\subsection{Deferred claims on the variance of the violation number measurements}
\label{section:deferred_claims}

In this section we give the proofs of the lemmas \cref{lemma:second_moment_Xs} and \cref{lem:2ndmoment}.

\subsubsection{Proof of \cref{lemma:second_moment_Xs}}
\label{section:proof_second_moment_xs}

As we explained just below \cref{definition:weighted_violation}, \cref{lemma:second_moment_Xs} argues that if the new (amplified) clause $f:[t]\rightarrow [m]$ is sampled using a random walk, its variance is similar to what it would be if $f$ were pairwise independent. It might be tempting to try to prove this claim using the fact that random walks on expanders are approximately pairwise independent (for $t = \Omega(\log m)$). Unfortunately, however, the walk length $t$ is for us only a constant, and so quite far from the regime of approximate pairwise independence.

Of course, a na\"ive version of Dinur's analysis in \cite[Section 6.2]{dinur2007pcp} would also encounter this obstacle. Instead, she uses only a `set-avoiding' property of random walks on expander graphs: For any fixed sets $A, B \subseteq [m]$ of no more than $\delta \cdot m$ in size, the probability that a random walk that starts in $A$ ends up in $B$ after $s$ steps is bounded above by $\delta + \mu^{s}$, where $\mu$ is the second largest eigenvalue of the expander graph. We refer the reader to our overview in \cref{section:techniques} (classical sequential repetition, and Technique 1: Commuting Layers), for an outline of how this property is useful in establishing the classical analog of our argument.

If all the terms in our original Hamiltonian $H$ were to commute, then to some extent one can reduce our analysis to the classical case: we simply imagine measuring $\sigma$ in a complete basis that diagonalises all the terms in $H$ simultaneously. This collapses the state to a mixture over product states, which we can deal with using the classical argument and convexity. Of course, it is not true that all the terms in $H$ commute, but we are able to proceed with the analysis by partitioning the terms in $H$ into at most constantly many commuting layers, amplifying and analysing each layer separately, and recombining the layers at the end up to some loss. We now present the proof of \cref{lemma:second_moment_Xs}.

\begin{lemma}[restatement of \cref{lemma:second_moment_Xs}] There exists a constant $C_\mu\leq 2/(1-\mu)$ dependent just on the collection of expander graphs, such that for every state $\sigma$ on $2t$ $n$-qubit registers, 
    \begin{align}
        \E[X_S^2]\leq (1+C_\mu)\sum_{j\in S}\Tr[H_{\reg j}\sigma] + \omega_{\min}\cdot  \sum_{i\neq j\in S}\Tr[H_{\reg i}\otimes H_{\reg j}\sigma]
    \end{align}
\end{lemma}
\begin{proof}
We wish to find an expression for the second moment
\begin{equation}
     \E[X_S^2]=\sum_\chi w_\chi\cdot  \mathbb{E}_{f\in \mathcal{F}_\chi} \Tr[(N_{f, S}^\chi)^2\sigma] 
\end{equation}

\noindent Let us focus on the square $(N_{f, S}^\chi)^2$ of a specific color $\chi$:

\begin{equation}
    \Tr[(N_{f, S}^\chi)^2\sigma] = \sum_{j\in S}\Tr[(\Pi_{f(j)}^\chi)_{\reg j}\sigma] +\sum_{i\neq j\in S} \Tr[(\Pi_{f(i)}^\chi)_{\reg i}\otimes (\Pi_{f(j)}^\chi)_{\reg j}\sigma]
\end{equation}
In expectation over $f$, the first of two terms simply reproduces $\mathbb{E}_f \Tr[N_{f, S}^\chi\sigma]$. Let us now focus on the second:
\begin{equation}
    \E_f\Tr[(\Pi_{f(i)}^\chi)_{\reg i} (\Pi_{f(j)}^\chi)_{\reg j}\sigma] = \sum_{u, v\in [m_\chi]} \mathbb{P}_f[f(j) = u \text{ and }f(i) = v] \cdot \Tr[(\Pi_{v}^\chi)_{\reg i} (\Pi_{u}^\chi)_{\reg j}\sigma]
\end{equation}

Note however that the projections $\{\Pi^\chi_u\}$ all mutually commute, and commute with $H_\chi$. Thus, we can expand $\sigma$ in an eigenbasis $\{\ket{E_\alpha}\}_\alpha$ of $H_\chi$. In particular, $\sigma$ is a convex combination of states:
\begin{equation}
    \sigma = \sum_\phi p_\phi \phi, \quad \ket{\phi} = \sum_{\Vec{\alpha} = \alpha_1, \cdots \alpha_{2t}}c_{\Vec{\alpha}} \otimes_{j\in [2t]}\ket{E_{\alpha_j}}, \text{ where } \sum p_{\phi}=1 \text{ and }\sum_{\vec{\alpha}} |c_{\vec{\alpha}}|^2=1.
\end{equation}

\noindent Then, 

\begin{equation}
    \Tr[(\Pi_{v}^\chi)_{\reg i} (\Pi_{u}^\chi)_{\reg j}\sigma] = \sum_\phi p_{\phi} \sum_{\vec{\alpha}} |c_{\vec{\alpha}}|^2 \bra{E_{\alpha_i}}\Pi_v^\chi\ket{E_{\alpha_i}}\cdot \bra{E_{\alpha_j}}\Pi_u^\chi\ket{E_{\alpha_j}}
\end{equation}

Crucially one can now re-express the expectation as a convex combination of quadratic forms, where for convenience we define the vectors 
\begin{gather}
    y^\alpha = (y^\alpha_1, \cdots, y^\alpha_{m_\chi})^T, \quad y^\alpha_u = \bra{E_{\alpha}}\Pi_v^\chi\ket{E_{\alpha}} 
\end{gather}
Recall $P_\chi$ is the transition matrix of a random walk on $G_\chi$. Indeed, note
\begin{align}
 &\sum_{u, v\in [m_\chi]} \mathbb{P}_f[f(j) = u \text{ and }f(i) = v] \cdot \Tr[(\Pi_{v}^\chi)_{\reg j} (\Pi_{u}^\chi)_{\reg j}\sigma] = \\
 &\sum_\phi p_{\phi} \sum_{\vec{\alpha}} |c_{\vec{\alpha}}|^2 \sum_{u, v} \mathbb{P}_f[f(j) = u \text{ and }f(i) = v] \cdot y^{\alpha_i}_v\cdot y^{\alpha_j}_u = \frac{1}{m_\chi} \sum_\phi p_{\phi} \sum_{\vec{\alpha}} |c_{\vec{\alpha}}|^2 \bigg( y^{\alpha_i} P^{|j-i|}_\chi y^{\alpha_j}\bigg)
\end{align}

\noindent From \cref{lemma:expander-quadratic}, we have 
\begin{align}
    y^{\alpha_i} P^{|j-i|}_\chi y^{\alpha_j} \leq \frac{1}{m_\chi} \|y^{\alpha_i}\|_1\cdot  \|y^{\alpha_j}\|_1 +\mu^{|j-i|} (\|y^{\alpha_i}\|_1+  \|y^{\alpha_j}\|_1)
\end{align}

\noindent We can now regroup these terms into $\sigma$, since each entry of $y^\alpha$ is positive:
\begin{gather}
    \frac{1}{m_\chi}\sum_\phi p_{\phi} \sum_{\vec{\alpha}} |c_{\vec{\alpha}}|^2 \|y^{\alpha_i}\|_1 = \frac{1}{m} \sum_\phi p_{\phi} \sum_{\vec{\alpha}} |c_{\vec{\alpha}}|^2 \sum_u \bra{E_{\alpha_i}} \Pi_v^\chi\ket{E_{\alpha_i}} = \Tr[(H^\chi)_{\reg i} \sigma] \\
    \frac{1}{m^2_\chi}\sum_\phi p_{\phi} \sum_{\vec{\alpha}} |c_{\vec{\alpha}}|^2 \|y^{\alpha_i}\|_1\|y^{\alpha_2}\|_1 = \frac{1}{m^2} \sum_\phi p_{\phi} \sum_{\vec{\alpha}} |c_{\vec{\alpha}}|^2 \sum_{u, v} \bra{E_{\alpha_i}} \Pi_v^\chi\ket{E_{\alpha_i}}\bra{E_{\alpha_j}} \Pi_u^\chi\ket{E_{\alpha_j}} = \\
    = \Tr[(H_\chi)_{\reg i}\otimes (H_\chi)_{\reg j} \sigma]
\end{gather}

\noindent We can now conclude by re-grouping the colors, using $\omega_{\text{min}} = (\min_{\chi\in [g]} w_\chi)^{-1}$:
\begin{gather}
    \sum_{\chi\in [g]} w_\chi\cdot \Tr[(H_\chi)_{\reg i}\otimes (H_\chi)_{\reg j} \sigma] \leq \sum_{\chi\in [g]}  \Tr[H_{\reg i}\otimes (H_\chi)_{\reg j} \sigma] \\
    \leq \omega_{\min} \sum_{\chi\in [g]} w_\chi\cdot  \Tr[H_{\reg i}\otimes (H_\chi)_{\reg j} \sigma] \leq \omega_{\min} \cdot  \Tr[H_{\reg i}\otimes H_{\reg j} \sigma].
\end{gather}
With $C_\mu =\max_i \sum_{j\neq i} \mu^{|j-i|} \leq \frac{2}{1-\mu}$ we conclude the proof. 
\end{proof}

\subsubsection{Proof of \cref{lem:2ndmoment}}

As explained at the start of \cref{sec:low-energy-analysis}, in this section (`the low-energy case') we complete the analysis of the moments of the color-weighted violation number variable $X$ (cf. \cref{definition:weighted_violation} and \cref{def:random-variable-X}). $X$ captures (for a random color $\chi$) the number of violated clauses in a random set (or path) of clauses determined by a function $f:[t]\rightarrow [m_\chi].$ In \cref{section:proof_second_moment_xs}, we proved \cref{lemma:second_moment_Xs}, which reduced the analysis in the case where $f$ is defined using expander walks to the case where $f$ is sampled from a pairwise independent function family. It remains to prove \cref{lem:2ndmoment}, which deals with the pairwise independent case. 

We view \cref{lem:2ndmoment} as a \emph{decoupling} statement. It says, very loosely speaking, that pairwise independent sampling is `sufficiently decoupling' in that we can get good (i.e. similar to classical) bounds on the second moment of the `violation number' measurement we are interested in, even when the measurement is done on a potentially entangled state.

\paragraph{The analogous classical analysis.}  For pairwise independent $f$ and any state $\tau$,
\begin{align}
    \E[X^2] = \sum_{j\in S}\Tr[H_{\reg j}\tau] +   \sum_{i\neq j\in S}\Tr[H_{\reg i}\otimes H_{\reg j}\tau] \label{eq:second-moment-ideal}
\end{align}

\noindent If $\tau$ were a product state, then after writing $\E[X^2]$ in the form of \cref{eq:second-moment-ideal} using pairwise independence, we trivially get
\begin{align}
    \sum_{i\neq j\in S}\Tr[H_{\reg i}\otimes H_{\reg j}\tau] = \sum_{i\neq j\in S}\Tr[H_{\reg i}\tau] \cdot \Tr[H_{\reg j}\tau] \leq \big(\sum_{i\in S}\Tr[H_{\reg i}\tau]\big)^2 \label{eq:ideal-decoupling}
\end{align}
because the tensor product decouples if $\tau$ is product. Hence, we succeed in bounding the variance of $X$:
\begin{equation}
    \E[X] = \sum_{j\in S}\Tr[H_{\reg j}\tau] \Rightarrow \E[X^2]\leq \E[X] + \E[X]^2
\end{equation}
The reader can verify that plugging this (and the elementary $\E[X] \geq t \cdot \lambda_{\min}(H)$) into \cref{fact:smm} then yields
\begin{align}
    \Pr[X > 0] \geq \frac{\E[X]^2}{\E[X]^2 + \E[X]} \geq \min\{\frac{t}{2} \cdot \lambda_{\min}(H) \:,\: \frac{1}{2}\}.
\end{align}

\paragraph{de Finetti motivation.} If our goal is to imitate this classical argument quantumly, one natural---albeit perhaps over-optimistic---starting point would be to attempt to reduce an entangled $\tau$ to a convex combination of product states. Quantum information theory gives us a class of tools for doing this, in the form of the \emph{quantum de Finetti theorems},  which roughly speaking capture the fact that a random $k$ qudit marginal of any $t\gg k$ qudit state should be close to separable. Unfortunately, we cannot hope to rely on a quantum de Finetti theorem as a black box. In particular, the best possible de Finetti theorems require an unviably large $t = \Omega(\log d)$, where $d$ is the dimension of a single qudit  (which could be exponential in the size $n$ of the original Hamiltonian).
However, perhaps surprisingly, we are able to progress with a technique that is inspired by a particular style of proof for a de Finetti theorem \cite{brandao2013quantum,vidick2016simple}.

The insight which allows us to prove \cref{lem:2ndmoment} is that we don't actually care if the state $\tau$ is product or not: separability is a stringent constraint, and instead it is sufficient for us that $\tau$ `looks product' with respect to measurements of $H$, the original Hamiltonian. We could, for example, measure $\tau$ in an eigenbasis that diagonalises $H^{\otimes t}$, and then use a \emph{classical} de Finetti theorem on the measurement outcomes, since \cref{lem:2ndmoment} cares only about the trace of $H$ against $\tau$ (in various registers). It turns out that this approach also requires an unviably large $t$ (since the guarantees of classical de Finetti theorems also scale unfavourably with the alphabet size of the random variables being permuted)---but we can do even better by realising that we don't need to know the precise eigenvalues of the eigenstates we get in each register after the eigenbasis measurement. Indeed, a very coarse approximation (`is the eigenvalue big or small?') will suffice. This coarse-graining, which allows us to reduce the alphabet size of the measurement outcomes to which we apply de Finetti--inspired techniques, is the chief reason that we can reduce $t$ to something manageable.

This is the idea behind the auxiliary measurement we introduced in \cref{section:aux_energy_meas}. The binary projective measurements $\{\Pi^{< \alpha},\mathbb{I}-\Pi^{< \alpha}\}$ (\cref{definition:aux_meas}) act on the same Hilbert space as $H$, and $\Pi^{< \alpha}$ simply projects onto the energy eigenspaces of $H$ with lower energy than $\alpha$: therefore, $\Pi^{<\alpha}$ commutes with $H$. Measuring $\Pi^{< \alpha}$ on a random set of $t$ registers allows us to estimate the ``energy" $E$ of $\tau$ (the estimate we use is the number of $\geq \alpha$ outcomes). We then  measure the complementary $t$ registers, and use Chernoff-bound-like tools to argue that they are likely to land in eigenspaces with similar energy $E$. After that, our proof proceeds via a careful case analysis, which hinges on whether $E$ is high (\cref{eqn:dVS_final}) or low (\cref{eqn:dnotVS_final}). In effect, our strategy is to partition the Hilbert space in which $\tau$ lives into (constantly many) subspaces, each of which has predictable behaviour with respect to $H$; and we are able to proceed with the analysis after we project (‘pinch’) $\tau$ into one of these subspaces because the ‘pinching’ measurement commutes with the measurements of $H$.

\paragraph{A guide to the proof of \cref{lem:2ndmoment}.} We now give an overview of the proof of \cref{lem:2ndmoment}, and explain the role of the technical condition \cref{eqn:trivial_single_copy_bound}. Our goal is to upper bound the second moment of $X$, i.e. the second moment of the measurement operator $N_{f, S}^{\chi}$ (on average over $S,\chi$ and $f$---see \cref{def:colored-subset-violation-number-meas}) on the `primary registers' $S$, \textit{conditioned} on the `low-energy' outcome for the auxiliary measurement on the auxiliary registers $\bar{S}$. Since $f$ is assumed to be pairwise independent, we are more explicitly trying to upper bound the following trace,
\begin{align}
    \sum_{i\neq j \in S} \Tr[H_{\reg i} \otimes H_{\reg j} \rho_{\mathsf{meas}(S,\bar S=\mathsf{low})}], \label{eq:double-H-post-meas}
\end{align}
where the state $\rho_{\mathsf{meas}(S,\bar S=\mathsf{low})}$ denotes the post-measurement subnormalized\footnote{For notational convenience in this exposition, we work with subnormalised states; in the proof we spell out the probabilities explicitly.} state after the auxiliary measurement has been performed \emph{both in $\bar S$ and in $S$}, conditioned on a low-energy outcome in $\bar S$. 

\begin{remark}
It suffices to upper bound the trace of $H_{\reg i} \otimes H_{\reg j}$ on $\rho_{\mathsf{meas}(S,\bar S=\mathsf{low})}$ (rather than $\rho_{\mathsf{meas}(\bar S=\mathsf{low})}$, which is the state $N_{f, S}^{\chi}$ is supposed to be measured on according to the definition of $X$) because, by construction, the auxiliary measurement commutes with $H$ on every register.
\end{remark}
We analyse \cref{eq:double-H-post-meas} by splitting into two cases depending on whether the measurement result on the registers in $S$ is `high-energy' or `low-energy'. More specifically, we separately analyse
\begin{align}
\sum_{i\neq j \in S} \Tr[H_{\reg i} \otimes H_{\reg j} \rho_{\mathsf{meas}(S=\mathsf{high},\bar S=\mathsf{low})}] \quad \text{ and } \quad \sum_{i\neq j \in S} \Tr[H_{\reg i} \otimes H_{\reg j} \rho_{\mathsf{meas}(S=\mathsf{low},\bar S=\mathsf{low})}]
\label{eq:double-H-post-meas-high} 
\end{align}
and then we put the two bounds together using the law of total expectation. We now elaborate on the nature of the case division.\\

\noindent \paragraph{The High Energy Case: $S = \mathsf{high}, \bar S = \mathsf{low}$.}
This case happens when the measurement in $S$ has a high-energy outcome, even though we condition on a low-energy outcome from the measurement in $\bar S$. Intuition would suggest that this case is unlikely; however, a delicate problem with the `error term' in the Chernoff-style bound from \cref{lemma:de-finetti} arises, which makes the technical condition \cref{eqn:trivial_single_copy_bound} necessary. 

Indeed, a na\"ive analysis leveraging \cref{lemma:de-finetti} would easily be able to upper bound all the contributions to $\E[X^2]$ arising from the case $S = \mathsf{high}$ as follows:
\begin{gather}
    \Tr[\rho_{\mathsf{meas}(S=\mathsf{high},\bar S=\mathsf{low})}] \leq e^{-\Omega(r^2/t)} \times \Tr[\rho_{\mathsf{meas}(\bar S=\mathsf{low})}],
\end{gather}
\noindent which entails
\begin{gather}\label{eq:68_overview_high}
    \sum_{i\neq j \in S} \Tr[H_{\reg i} \otimes H_{\reg j} \rho_{\mathsf{meas}(S=\mathsf{high},\bar S=\mathsf{low})}] \leq t^2 \times e^{-\Omega(r^2/t)} \times \Tr[\rho_{\mathsf{meas}(\bar S=\mathsf{low})}]. 
\end{gather}

\noindent Unfortunately, this analysis is too loose. Recall that we are trying to upper bound $\E[X^2]$ (and, by extension, the left-hand-side of \cref{eq:68_overview_high}) by some quantity related to $\E[X]$. We expect $\E[X]$ to be similar to $t \cdot \lambda_{\mathrm{min}}(H)$ (see \cref{lemma:expectation_XS}), and so $\E[X]$ may be inverse-polynomial. However, if we are in the low-energy case for the auxiliary $\bar S$ measurement, the right-hand-side of \cref{eq:68_overview_high} is a constant (for constant $t$). As such, the \textit{relative error} in our upper bound of $\E[X^2]$ in terms of $\E[X]$ is exceedingly large. 

The actual bound we aim for is roughly
\begin{align}\label{eq:68_overview_relative}
   (\text{\cref{eq:68_overview_high}, LHS}) \leq t \times e^{-\Omega(r^2/t)} \times \sum_{i\in S} \Tr[H_{\reg i}\rho_{\mathsf{meas}(\bar S=\mathsf{low})}].
\end{align}

\noindent Note that $\Tr[H_{\reg i}\rho_{\mathsf{meas}(\bar S=\mathsf{low})}]$ is exactly $\E[X]$ (conditioned on $S = \mathsf{low}$), and so this bound is good regardless of how large or small $\E[X]$ is relative to $t$. We explain briefly how we acheve this bound, and why it makes the technical condition \cref{eqn:trivial_single_copy_bound} necessary.

The `na\"ive analysis' which we mentioned earlier achieves \cref{eq:68_overview_high} purely by using the Chernoff-like bound from \cref{lemma:de-finetti} to upper bound the normalisation of $\rho_{\mathsf{meas}(S=\mathsf{high},\bar S=\mathsf{low})}$ in terms of that of $\rho_{\mathsf{meas}(\bar S=\mathsf{low})}$. However, \cref{eq:double-H-post-meas-high} also involves the original Hamiltonian $H$, and we would like to take advantage of this fact to get a bound that looks more like \cref{eq:68_overview_relative}. The main obstacle to this intuition is that the state we are tracing $H$ against in \cref{eq:double-H-post-meas-high} is not $\rho_{\mathsf{meas}(\bar S=\mathsf{low})}$ (which would be the most convenient if we want a bound in terms of $\E[X]$, since that is the state which appears in the definition of $X$), nor even $\rho_{\mathsf{meas}(S,\bar S=\mathsf{low})}$, but $\rho_{\mathsf{meas}(S=\mathsf{high},\bar S=\mathsf{low})}$, which involves an additional layer of conditioning.

Two observations allow us to make progress:
\begin{enumerate}
\item It is difficult to control the trace of $H$ against an arbitrarily conditioned state, but \cref{lemma:de-finetti} \emph{does not care} which state it is applied on, and
\item We can switch the order of the $H$ measurement and the auxiliary measurement because they commute.
\end{enumerate} Therefore, our approach is to imagine performing $H$ first and the auxiliary measurement second. With this approach, $H$ is performed on $\rho$ itself (later chosen as a ground state of the amplified Hamiltonian $H^{(2t)}$), and \cref{lemma:de-finetti} is then applied to the post-measurement state after the $H$ measurement. We then observe that, if $\Tr[H_{\reg i} \rho]$ for any $i$ is larger than $\E[X] \approx t \cdot \lambda_{\min}(H)$, we are already done lower bounding the amplified ground state energy. We therefore assume that $\Tr[H_{\reg i} \rho]$ is upper bounded by $\E[X]$ for all $i$ (which is the technical condition in \cref{eqn:trivial_single_copy_bound}). Together with \cref{lemma:de-finetti}, we show that this gives the desired bound in \cref{eq:68_overview_relative}.\\

\noindent \paragraph{The Low Energy Case: $S = \mathsf{low}, \bar S = \mathsf{low}$.} This case is comparatively straightforward, since this case is where we finally reap the fruits of our `miser's de Finetti' technique, and all the work involved in setting up its use has already been done. We have reached the branch of our double dichotomy (over the $\bar S$ and $S$ measurements) where the energy of $H$ in $S$ is likely to be `low' (upper bounded by $\alpha$); therefore, the energy of $H \otimes H$ is likely to be upper bounded by that of $\alpha \cdot \mathbb{I} \otimes H$, which is precisely $\alpha \cdot \E[X]$ (when $H$ is measured on the state in the definition of $X$). Roughly speaking, this case is where the `decoupling' effect of our de-Finetti-inspired technique can be seen: by making a scalar out of one of the copies of $H$ in $H \otimes H$ (which was possible because, during our `pinching' measurement, we projected that register into the low-energy subspace of $H$), it allows $H \otimes H$ to be bounded in terms of a quantity that only involves a single $H$, which is easy to relate to the mean of $X$.

We now present the formal statement and proof of \cref{lem:2ndmoment}.

\begin{lemma}[restatement of \cref{lem:2ndmoment}]
Suppose that $\rho$ is a state such that for all $i$, 
\begin{align}
\tr{H_{\reg i} \rho} \leq \E_{|S|=t} \sum_{j \in S} \sum_{c\in \overline U_{\overline S}} \prs{\rho}{c} \ \tr{H_{\reg j} \rho_{|S,\alpha,c}} \,. \label{eqn:trivial_single_copy_bound-2}
\end{align}
Then for any choice of parameters $\alpha > 0$ and $r \in [t]$ it holds that
\begin{align*}
\E_{|S|=t} \sum_{c\in \overline U_{\overline S}} \prs{\rho}{c} \sum_{\substack{i,j \in S \\ i \neq j}} \tr{H_{\reg i} \cdot H_{\reg j} \rho_{|S, \alpha, c} }
\leq (8r + \alpha t + 2t \cdot e^{-\frac{8r^2}{t}}) \cdot \E_{|S|=t} \sum_{j \in S} \sum_{c\in \overline U_{\overline S}} \prs{\rho}{c} \tr{H_{\reg j} \rho_{|S, \alpha, c} } \,.
\end{align*}
Here, $U_{\overline S}$ is defined as in \cref{eqn:us_def} and depends implicitly on $\alpha$ and $r$.
\end{lemma}

\begin{proof}

For any $i, j, k \in [2t]$, the commutator $[H_{\reg i} \cdot H_{\reg j}, \Pi^{\geq \alpha}_{\reg k}] = 0$.
Note that this still holds if $i=k$ or $j=k$ because $\Pi^{\geq \alpha}_{\reg k}$ projects onto a direct sum of eigenspaces of $H_{\reg k}$.
Operationally, this means that we can first perform the measurement $\{\1 - \Pi^{\geq \alpha}_{\reg k}, \Pi^{\geq \alpha}_{\reg k}\}$ on all registers $k$ without altering the measurement outcome of measuring $H_{\reg i} \cdot H_{\reg j}$.
In other words, if we are only interested in measuring $H_{\reg i} \cdot H_{\reg j}$, we can first perform the auxiliary energy measurement \emph{on all registers}.

More formally, we denote by $\rho_{|S,\alpha,c,d}$ the post-measurement state after performing the additional auxiliary energy measurement on the $S$-registers of $\rho_{|S,\alpha,c}$ and conditioning on receiving outcome $d \in \bits^S$, and denote by $\prs{\rho_{|S,\alpha,c}}{d}$ the probability of this outcome.
Then 
\begin{align}
&\E_{|S|=t} \sum_{c\in \overline U_{\overline S}} \prs{\rho}{c} \sum_{\substack{i,j \in S \\ i \neq j}} \tr{H_{\reg i} \cdot H_{\reg j} \rho_{|S, \alpha, c} } =
\\
=& \E_{|S|=t} \sum_{c\in \overline U_{\overline S}} \sum_{d \in \bits^S} \prs{\rho}{c} \prs{\rho_{|S,\alpha,c}}{d} \sum_{\substack{i,j \in S \\ i \neq j}} \tr{H_{\reg i} \cdot H_{\reg j} \rho_{|S, \alpha, c, d} } \,. \label{eqn:int0}
\end{align}

We define the set $V_S = \{d \in \bits^S \;:\; |d| \geq 8r\}$ and split the sum over $d$ into two sums, one over $d \in V_S$ and one over $d \in \overline V_S$.
We will bound each sum in turn.

\paragraph{Bounding the sum over $d \in V_S$.} Let $\Pi_{S, \alpha}^{(c, d)}$ be the projector corresponding to performing the auxiliary energy measurement on all registers and receiving outcome $c$ on registers in $\overline{S}$ and outcome $d$ on registers in $S$.
Then, 
\begin{align*}
\prs{\rho}{c} \prs{\rho_{|S,\alpha,c}}{d} \rho_{|S, \alpha, c, d} = \Pi_{S, \alpha}^{(c, d)} \rho \Pi_{S, \alpha}^{(c, d)} \,.
\end{align*}
We write the spectral decomposition of $H$ as 
\begin{align*}
H = \sum_{a} \lambda_a \Pi^{(a)}
\end{align*}
for eigenvalues $\lambda_a \geq 0$ and orthogonal projectors $\Pi^{(a)}$.
Since $\Pi^{\geq \alpha}_{\reg k}$ is a sum of the projectors $\Pi^{(a)}_{\reg k}$, $[\Pi^{(a)}_{\reg k}, \Pi_{S, \alpha}^{(c, d)}] = 0$, and trivially $[\Pi^{(a)}_{\reg i}, \Pi^{(b)}_{\reg j}] = 0$ for $i \neq j$.
Therefore, we can bound the $V_S$-part of the sum in \cref{eqn:int0} as 
\begin{align*}
&\E_{|S|=t} \sum_{c\in \overline U_{\overline S}} \sum_{d \in V_S} \prs{\rho}{c} \prs{\rho_{|S,\alpha,c}}{d} \sum_{\substack{i,j \in S \\ i \neq j}} \tr{H_{\reg i} \cdot H_{\reg j} \rho_{|S, \alpha, c, d} } \\
&\leq t \cdot \E_{|S|=t} \sum_{c\in \overline U_{\overline S}} \sum_{d \in V_S} \prs{\rho}{c} \prs{\rho_{|S,\alpha,c}}{d} \sum_{i \in S} \tr{H_{\reg i} \rho_{|S, \alpha, c, d} } \\
&= t \cdot \E_{|S|=t} \sum_{c\in \overline U_{\overline S}} \sum_{d \in V_S} \sum_{i \in S} \sum_{a} \lambda_{a} \tr{\Pi_{\reg i}^{(a)} \Pi_{S, \alpha}^{(c, d)} \rho \Pi_{S, \alpha}^{(c, d)} } \\
&= t \cdot \E_{|S|=t} \sum_{c\in \overline U_{\overline S}} \sum_{d \in V_S} \sum_{i \in S} \sum_{a} \lambda_{a} \mathrm{Tr}\bigg[ \Pi_{S, \alpha}^{(c, d)} \underbrace{\left(\Pi_{\reg i}^{(a)}  \rho \Pi_{\reg i}^{(a)} \right)}_{\coloneqq \rho_{i,a}} \bigg] \\
&= t \cdot \E_{|S|=t} \sum_{c\in \overline U_{\overline S}} \sum_{d \in V_S} \sum_{i \in S} \sum_{a} \lambda_{a} \tr{\rho_{i,a}} \sum_{c\in \overline U_{\overline S}} \sum_{d \in V_S} \tr{ \Pi_{S, \alpha}^{(c, d)} \rho_{|i,a} } \,. \numberthis \label{eqn:intexp}
\end{align*}
In the last line, we defined the renormalised conditioned state 
\begin{align*}
\rho_{|i,a} = \frac{\rho_{i,a}}{\tr{\rho_{i,a}}} \,.
\end{align*}
We now observe that if we extend the sum from $i \in S$ to all $i \in [2t]$, the expression in \cref{eqn:intexp} can only increase since each term is non-negative.
Therefore, we can bound 
\begin{align*}
\text{\cref{eqn:intexp}} 
&\leq \E_{|S|=t} \sum_{i} \sum_{a} \lambda_{a} \tr{\rho_{i,a}} \sum_{c\in \overline U_{\overline S}} \sum_{d \in V_S} \tr{ \Pi_{S, \alpha}^{(c, d)} \rho_{|i,a} } \\
&= \sum_{i} \sum_{a} \lambda_{a} \tr{\rho_{i,a}} \left( \E_{|S|=t}  \sum_{c\in \overline U_{\overline S}} \sum_{d \in V_S} \tr{ \Pi_{S, \alpha}^{(c, d)} \rho_{|i,a} } \right) \numberthis \label{eqn:intexp1}
\end{align*}

Here, we can apply the de Finetti reasoning once again. The expression in parenthesis captures the probability of recieving a low energy string on $\Bar{S}$ and a high-energy string on $S$. From \cref{lemma:de-finetti}, we have for any state $\sigma$,
\begin{align*}
\E_{|S|=t} \sum_{c\in \overline U_{\overline S}} \sum_{d \in V_S} \tr{ \Pi_{S, \alpha}^{(c, d)} \sigma } \leq e^{-\frac{8r^2}{t}}\,. 
\end{align*}

\noindent Inserting this into \cref{eqn:intexp1} and then re-inserting the definitions of $\rho_{i,a}$ and $H_{\reg i}$, we get 
\begin{align*}
\text{\cref{eqn:intexp1}} \leq e^{-\frac{8r^2}{t}} \; \sum_{i} \tr{H_{\reg i} \rho} \,.
\end{align*}
Using the assumption in \cref{eqn:trivial_single_copy_bound} and combining all the steps, we get 
\begin{align}
\E_{|S|=t} \sum_{c\in \overline U_{\overline S}} \sum_{d \in V_S} \prs{\rho}{c} \prs{\rho_{|S,\alpha,c}}{d} \sum_{\substack{i,j \in S \\ i \neq j}} \tr{H_{\reg i} \cdot H_{\reg j} \rho_{|S, \alpha, c, d} } \leq e^{-\frac{8r^2}{t}}\cdot 2t \cdot \E_{|S|=t} \sum_{j \in S} \sum_{c\in \overline U_{\overline S}} \prs{\rho}{c} \ \tr{H_{\reg j} \rho_{|S,\alpha,c}} \,. \label{eqn:dVS_final}
\end{align}

\paragraph{Bounding the sum over $d \in \overline V_S$.}
Next, we need to bound $\tr{H_{\reg i} \cdot H_{\reg j} \rho_{|S, \alpha, c, d} }$ for $c \in \overline U_{\overline S}$ and $d \in \overline V_S$.
For this, we fix $i \neq j$ and distinguish two cases.
If $d_i = 1$ (i.e.~the auxiliary energy measurement on register $i$ yielded outcome ``$\Pi^{\geq \alpha}_{\reg i}$''), then we use the trivial bound $H_{\reg i} \leq \1$ (and the fact that $H_{\reg i}$ and  $H_{\reg j}$ commute since $i \neq j$) to get 
\begin{align*}
\tr{H_{\reg i} \cdot H_{\reg j} \rho_{|S, \alpha, c, d} }
\leq \tr{H_{\reg j} \rho_{|S, \alpha, c, d} } \,.
\end{align*}
On the other hand, if $d_i = 0$ then we use the fact that $\rho_{|S, \alpha, c, d} = \Pi^{<\alpha}_{\reg i} \rho_{|S, \alpha, c, d} \Pi^{<\alpha}_{\reg i}$, $[H_{\reg j}, \Pi^{<\alpha}_{\reg i}]=0$, and $\Pi^{<\alpha}_{\reg i} H_{\reg i} \Pi^{<\alpha}_{\reg i} \leq \alpha \1$ to bound 
\begin{align*}
&\tr{H_{\reg i} \cdot H_{\reg j} \rho_{|S, \alpha, c, d} } 
= \tr{H_{\reg i} \cdot H_{\reg j} (\Pi^{<\alpha}_{\reg i} \rho_{|S, \alpha, c, d} \Pi^{<\alpha}_{\reg i}) } = \\ &=\tr{(\Pi^{<\alpha}_{\reg i} H_{\reg i} \Pi^{<\alpha}_{\reg i}) \cdot H_{\reg j} \rho_{|S, \alpha, c, d} } \leq \alpha \tr{H_{\reg j} \rho_{|S, \alpha, c, d} } \,.
\end{align*}
By definition of $\overline V_S$, there can be at most $8r$ indices $i \in S$ for which $d_i = 1$.
Therefore
\begin{align}
\sum_{\substack{i,j \in S \\ i \neq j}} \tr{H_{\reg i} \cdot H_{\reg j} \rho_{|S, \alpha, c, d} } \leq (8r + \alpha t) \sum_{j \in S} \tr{H_{\reg j} \rho_{|S, \alpha, c, d} } \,. \label{eqn:int2}
\end{align}
Consequently, 
\begin{align*}
&\E_{|S|=t} \sum_{c\in \overline U_{\overline S}} \sum_{d \in V_S} \prs{\rho}{c} \prs{\rho_{|S,\alpha,c}}{d} \sum_{\substack{i,j \in S \\ i \neq j}} \tr{H_{\reg i} \cdot H_{\reg j} \rho_{|S, \alpha, c, d} } \\
&\leq \E_{|S|=t} \sum_{c\in \overline U_{\overline S}} \sum_{d \in V_S} \prs{\rho}{c} \prs{\rho_{|S,\alpha,c}}{d} (8r + \alpha t) \sum_{j \in S} \tr{H_{\reg j} \rho_{|S, \alpha, c, d} } \\
&\leq (8r + \alpha t) \cdot \E_{|S|=t} \sum_{j \in S} \sum_{c\in \overline U_{\overline S}} \prs{\rho}{c} \tr{H_{\reg j} \rho_{|S, \alpha, c} } \,. \numberthis \label{eqn:dnotVS_final}
\end{align*}

\paragraph{Combining both bounds.}
We can now insert the bounds for the sum over $d \in V_S$ (\cref{eqn:dVS_final}) and the sum over $d \in \overline{V}_S$ (\cref{eqn:dnotVS_final}) into \cref{eqn:int0} to get 
\begin{align*}
\E_{|S|=t} \sum_{c\in \overline U_{\overline S}} \prs{\rho}{c} \sum_{\substack{i,j \in S \\ i \neq j}} \tr{H_{\reg i} \cdot H_{\reg j} \rho_{|S, \alpha, c} } \leq (8r + \alpha t + 2t \cdot e^{-\frac{8r^2}{t}}) \cdot \E_{|S|=t} \sum_{j \in S} \sum_{c\in \overline U_{\overline S}} \prs{\rho}{c} \tr{H_{\reg j} \rho_{|S, \alpha, c} } 
\end{align*}
as claimed.

\end{proof}

\subsection{A careful choice of parameters}
\label{section:parameters}

\begin{corollary} \label{cor:lowerbound_simpler}
Suppose that $10^5 \leq t$.
Then 
\begin{align*}
\lambda_{\min}(H^{(2t)}) \geq \min\Big\{ \frac{1}{3}\frac{\log t}{t} \:\:, \quad \lambda_{\min}(H) \times \underbrace{\frac{t^{1/2}}{(\log t)^{1/2}}}_{\text{amplification factor}} \times \underbrace{\frac{1}{20} \big[\max\{1 + C_\mu, \, \omega_{\min}\}\big]^{-1}}_{\text{constant not depending on $t$}} \Big\}
\end{align*}
\end{corollary}
\begin{proof}
For simplicity, we write $\gamma = \lambda_{\min}(H)$.

We use \cref{lem:lowerbound} with the following parameter choices: 
\begin{align*}
r &= (t \log t)^{1/2} \, \\
\alpha &= \frac{r}{t} \,.
\end{align*}
Note that this choice of parameters gives us $e^{-2r^2/t} = \frac{1}{t^2}$.

Suppose that $\Delta \geq 1/2$. Then the first term in the max from \cref{lem:lowerbound} comes to
\begin{align}
\frac{2 \alpha r}{t}\Big(\Delta - e^{-2 r^2/t}\Big)
&=
\frac{2 r^2}{t^2}\Big(\frac{1}{2} - \frac{1}{t^2}\Big) \\
&= \frac{2\log t}{t} \Big(\frac{1}{2} - \frac{1}{t^2}\Big) \\
&\geq \frac{1}{3}\frac{\log t}{t}
\end{align}
assuming $\frac{1}{t^2} \leq \frac{1}{6}$.

Now suppose that $\Delta \leq 1/2$. Then the second term in the max from \cref{lem:lowerbound} is at least
\begin{align}
&\frac{t}{2} \cdot \gamma \cdot \Big[ \max\{1 + C_\mu, \, \omega_{\min}\} \cdot (1 + 8r + \alpha t + 2t \cdot e^{-8r^2/t}) \Big]^{-1} \\
&\geq
\frac{t}{2} \cdot \gamma \cdot \Big[ \max\{1 + C_\mu, \, \omega_{\min}\} \cdot (1 + 9t^{1/2} \log^{1/2}t  + 2t \cdot \frac{1}{t^8}) \Big]^{-1} \\
&\geq \frac{t}{2} \cdot \gamma \cdot \Big[ \max\{1 + C_\mu, \, \omega_{\min}\} \cdot (10t^{1/2} (\log t)^{1/2} \Big]^{-1} \quad \\
&\geq \gamma \cdot \underbrace{\frac{t^{1/2}}{(\log t)^{1/2}}}_{\text{amplification factor}} \cdot \underbrace{\frac{1}{20} \big[\max\{1 + C_\mu, \, \omega_{\min}\}\big]^{-1}}_{\text{constant not depending on $t$}}.
\end{align}
\end{proof}

\section{Iterated Amplification}
\label{sec:iterated_to_constant}
We can use \cref{thm:tp_amplification_intro} repeatedly to get to a QMA-hard family of (non-local) Hamiltonians with constant promise gap. Our starting point is a family of local Hamiltonians which is ``equitable", in the sense that it can be partitioned into a collection of $g = O(1)$ commuting terms.

\begin{assumption}[$k$-$\mathsf{LH[a, b; \omega]}$]\label{assumption:balanced_H}
    Let $\{H_n\}_{n\in \mathbb{N}}$ be a family of $k$-$\mathsf{LH}$, where each $H_n$ is an expectation over $\poly(n)$ many $k$-local projections on $n$ qubits. Further, assume 
    \begin{enumerate}
        \item \textbf{Soundness-Completeness Gap.} There exists a negligible\footnote{We use the notation $\mathsf{negl}(n)$ (negligible) to denote $O(n^{-i})$ for every $i$.}  function $\mu(n)$ and a positive polynomial $p(n)$ such that 
        \begin{equation}
            \text{Either } \lambda_{min}(H_n)\leq \mu(n) = a, \text{ or } \lambda_{min}(H_n)\geq p(n) = b
        \end{equation}

        \item \textbf{Equitable Coloring.} The projections in $H_n$ can be partitioned into $g = O(1)$ commuting layers, wherein the weights in each layer is roughly balanced. In particular, there exists an explicit constant
        \begin{equation}
          (\cref{eq:min-weight-def})\text{ }\omega\equiv \omega_{min} = O(1)
        \end{equation}
    \end{enumerate}
\end{assumption}

\begin{remark}
    Since it is not known whether $\mathsf{QMA}=\mathsf{QMA}_1$ (one-sided error), we require our amplification to start from a Hamiltonian family where the lowest eigenvalue is promised to either be negligibly close to 0, or at least inverse polynomially far from 0.
\end{remark}

In the subsequent \cref{section:qma_complete_equitable} we prove there exists a family of $k$-$\mathsf{LH[2^{-\poly(n)}, 1/\poly(n); O(1)]}$ for which it is $\mathsf{QMA}$-complete to distinguish between the high and low energy cases.

\begin{theorem}
    [Iterated Amplification] \label{theorem:iterated_amplification}
    Assume  $\{H_n\}_{n\in \mathbb{N}}$ be a family of $k$-$\mathsf{LH}[\mu(n), p(n); \omega]$ satisfying \cref{assumption:balanced_H}. Then, there exists an explicit constant $c(\omega)$ and a deterministic construction of a family of amplified Hamiltonians $\{\tilde{H}_n\}_n$ with the following parameters:
    \begin{enumerate}
        \item \textbf{Instance Size.} Each $\tilde{H}_n$ is an expectation of $\poly(n)$ projections over $\poly(n)$ many qubits. The locality of each amplified projection is $k\cdot p(n)^{O(1)}$.

        \item \textbf{Amplified Promise Gap.} If $\lambda_{min}(H_n)\geq p(n)$ then $\lambda_{min}(\tilde{H}_n)\geq c(\omega)$; if $\lambda_{min}(H_n)\leq \mu(n)$ then $\lambda_{min}(\tilde{H}_n)\leq \mu(n)\cdot p(n)^{20}$.
    \end{enumerate}
\end{theorem}

\begin{proof}
We construct the family of Hamiltonians $\{\tilde H_n\}$ as follows. We will fix $n$ and not include it in subscripts for simplicity. We will fix an integer parameter $t$, and iteratively amplify to define a sequence of Hamiltonians $M_0, M_1\cdots M_\ell$:

\begin{equation}
    M_0 = H_n, \quad M_{i} = M_{i-1}^{(2t)} \text{ for } i\in [\ell]
\end{equation}

\noindent Note that upon each iteration, the weight parameter $\omega_{\min}$ is fixed. To simplify notation, let us fix the constant 

\begin{equation}
    \eta\equiv \frac{1}{20\cdot \max(1+C_\mu, \omega_{\min})}
\end{equation}
 Let us first consider the soundness. From the iterative application of \cref{cor:lowerbound_simpler}, we have 
\begin{equation}
    \lambda_{min}(M_\ell) \geq \text{min}\bigg[ \frac{1}{3}\frac{\log t}{t}, \quad \lambda_{min}(M_0)\times \bigg(\eta^2\cdot \frac{t}{\log t} \bigg)^{\ell/2} \bigg]
\end{equation}

\noindent To simplify the computation, we fix $t$ such that $t^{1/3}\geq \log t$ and $t > \eta^{-6}$, in addition to the conditions in \cref{cor:lowerbound_simpler}. Thereby, 

\begin{equation}
   \lambda_{min}(M_\ell) \geq \text{min}\bigg[ \frac{1}{3}\frac{\log t}{t}, \quad \lambda_{min}(M_0)\times  t^{\ell/6} \bigg]
\end{equation}

\noindent If $\lambda_{min}(H_n)\geq 1/p(n)$, then it suffices to pick the iteration $\ell$ to satisfy

\begin{equation}
  \ell = \bigg \lceil\frac{ 6\log  p(n) }{ \log t}\bigg\rceil \text{ to ensure } \lambda_{min}(M_\ell) \geq \frac{1}{3}\frac{\log t}{t}
\end{equation}

Under such a choice of $v$, we can now consider the completeness, number of qubits and clauses, and the locality of the amplified Hamiltonians. First, we observe that the number of qubits of $M_k$ is
\begin{equation}
    n\times (2t)^\ell \leq 2t\cdot n\cdot p(n)^{6\log 2t/\log t} \leq 2t n p(n)^{12}
\end{equation}

\noindent The number of clauses of $M_\ell$ is increased by a multiplicative factor of 

\begin{equation}
    d^{2t\cdot \ell} \leq d^{2t}\cdot p(n)^{12t/\log t}
\end{equation}
\noindent which is $\poly(n)$ so long as $t$ is a constant. Following the same calculation above, the locality of $M_\ell$ is bounded by 
\begin{equation}
    k\times (2t)^{\ell} \leq k\times 2t\times p(n)^{12}
\end{equation}

\noindent Finally, from \cref{lem:upperbound}, the completeness of the amplified Hamiltonian is given by
\begin{equation}
    \lambda_{min}(M_k)\leq (2t)^\ell \times \lambda_{min}(H_n) \leq 2t\mu(n)\cdot p(n)^{12}
\end{equation}

\end{proof}

We describe a simple consequence of this iterated amplification.

\begin{theorem}
    [A ``Streaming" Quantum PCP Theorem]
    There exists a family $\{H_n\}$ of Hamiltonians on $n$ qubits and an explicit constant $c$, wherein each term is an $O(n)$-fold tensor product of $O(1)$-local projections, such that it is $\mathsf{QMA}$-Complete to decide whether the ground state energy of $\{H_n\}$ is $\leq \mathsf{negl}(n)$ or $\geq c$.
\end{theorem}

\begin{proof}
    As the starting point to the amplification, we consider the family of ``equitable" local Hamiltonians ensured by \cref{corollary:equitable}, which has $\omega_{min} \geq 1/2\cdot 1/35$. \cref{theorem:iterated_amplification} then ensures the amplification up to constant soundness, while maintaining neglible completeness. It only remains to argue the containment in $\mathsf{QMA}$. To do so, we note that one can measure the energy of any candidate witness via phase estimation, for which it suffices (via trotterization) to argue that one can implement the Hamiltonian simulation of a tensor product of $O(n)$ $O(1)$-local projections, as guaranteed by \cref{lemma:ham_sim} below. 
\end{proof}

\begin{remark}
    Using a standard padding argument, the locality of the family of Hamiltonians can be reduced to $n^\epsilon$ for any constant $\epsilon$. 
\end{remark}

\begin{lemma}
    [Hamiltonian Simulation of Tensor-Product Hamiltonians]\label{lemma:ham_sim} Fix $T>0$, and let $\{\Pi_i\}_{i\in [a]}$ be a collection of $k$-local projections. Then, there exists a quantum circuit of $(k+1)$-local gates of depth $O(\log m)$ and size $O(m)$ which performs the Hamiltonian simulation $e^{i\Pi T}$ of the projection 
    \begin{equation}
        \Pi = \mathbb{I} - \bigotimes_i^m (\mathbb{I}-\Pi_i).
    \end{equation}
\end{lemma}

\begin{remark}
    The $k$-local gates can be simulated using $2$-local gates to arbitrary precision using the Solovay-Kitaev Theorem \cite{Nielsen_Chuang_2010}, up to cost exponential in $k$. 
\end{remark}

\begin{proof}
    To begin, note that one can coherently measure each clause $\Pi_{i}$ into an ancilla register initialized to $\ket{0}$, using a unitary $U_i$ on $k+1$ qubits which implements
    \begin{equation}
        U_{i} = \Pi_i\otimes  X  + (\mathbb{I} - \Pi_i)\otimes\mathbb{I}
    \end{equation}
    Subsequently we compute (coherently) the AND of all ancillas, using a tree of $\leq 2m$ ancillas and depth $\leq 1+\log m$. The root bit contains the outcome of a measurment of $\Pi$,  the tensor product of projections. One can now apply the single-qubit phase gate $\text{diag}(1, e^{iT})$, and subsequently uncompute all the ancillas, to concludes the circuit.
\end{proof}

\section{$\mathsf{QMA}$-Completeness of Equitable Local Hamiltonians}
\label{section:qma_complete_equitable}

\begin{claim}
    [Degree Reduction for FK Hamiltonians, \cite{Anshu2023CircuittoHamiltonianFT}]\label{claim:degree-reduction} Any $\mathsf{QMA}$ protocol involving an $n$-qubit verifier circuit $V$ with $T=\poly(n)$ two-qubit gates can be mapped into a $5-\mathsf{LH}[a, b]$ $H$ on $\poly(n)$ qubits with $a = 2^{-\poly(n)}$ and $b=a+1/\poly(n)$. Furthermore, 
    \begin{enumerate}
        \item each qubit is involved in at most 7 terms in the Hamiltonian.
        \item each Hamiltonian term is an unweighted projection.
    \end{enumerate}
\end{claim}

Consider the graph $G = ([m], E)$ defined on the $m$ clauses/terms of the Hamiltonian, where two clauses are connected if they overlap on a qubit. Note that the degree of $G$ is $\leq 7\cdot 5$.

\begin{definition}
    An \textit{equitable coloring} of a graph $G=(V, E)$ on $g$ colors is a partition of the vertex set $V = V_1\cup V_2\cdots \cup V_g$ into $g$ disjoint subsets such that no two adjacent vertices have the same color, and furthermore the number of vertices per subset is balanced; $\forall i, j\in [n]$: $||V_i|-|V_j||\leq 1$.
\end{definition}

\begin{theorem}
[\cite{KIERSTEAD_KOSTOCHKA_2008}]\label{theorem:equitable_coloring}
    For any graph $G$ on $n$ vertices of maximum degree $d$, there exists an efficient algorithm to find an equitable coloring of $G$ on $d+1$ colors in time $\poly(n)$.
\end{theorem}

\begin{corollary}
    [$\mathsf{QMA}$-complete Layered Hamiltonians]\label{corollary:equitable}
    In the context of \cref{claim:degree-reduction}, the clauses $C = \{h_i\}_i$ of the local Hamiltonian $H = \sum_{i\in C} h_i$ can be partitioned into $O(1)$ subsets, where the clauses within each subset commute and the sizes of the subsets differ by at most 1. 
\end{corollary}

\begin{proof}
    Consider the clauses graph $G$ of the family of Hamiltonians in \cref{claim:degree-reduction}. The equitable coloring of $G$ produced by \cref{theorem:equitable_coloring} gives the desired partition of the Hamiltonian terms. Since no two terms in the same subset overlap, they also commute. 
\end{proof}

\bibliography{main}

\begin{appendix}

\section{Amplification from the Detectability Lemma}
\label{section:dl}

In this section we present an alternative ``one-shot'' construction of a non-local Hamiltonian with constant gap, based on the detectability lemma and its converse.\footnote{We thank Anurag Anshu for pointing us to this construction.} Roughly speaking, for any integer $t$, we will define an amplification scheme which maps local Hamiltonians on $n$ qubits and $m$ clauses of (sufficiently small) inverse-polynomial promise gap $\gamma$, to a new Hamiltonian on $t\cdot n$ qubits and $O(1)$ clauses, with promise gap $\Theta(\gamma\cdot m \cdot t)$. While these results are conceptually similar to our derandomization scheme in \cref{theorem:lowerbound_simpler}, the new Hamiltonian will be $O(n\cdot t)$ local (i.e. fully global, see \cref{theorem:amp-from-dl} below). We refer the reader back to \cref{section:discussion} for a discussion.

We begin with a brief background on the detectability lemma and its converse.

\begin{lemma}
    [Detectability Lemma, e.g. {\cite[Lemma 2]{AAV16}}]\label{lemma:DL} Let $H = \mathbb{E}_{i\in[m]}\Pi_i$, where each $\Pi_i$ is a projection that commutes with all except at most $d$ of the other projections $\Pi_j$. Then, if $\lambda_{\mathsf{min}}(H)\geq \gamma$, we have for every state $\ket{\psi}$:
    \begin{equation}
        \big\|\prod_{i\in [m]} \big(\mathbb{I}-\Pi_{i}\big)\ket{\psi}\big\|^2 \leq (1+m\gamma/d^2)^{-1}\;.
    \end{equation}
\end{lemma}

\begin{lemma}
    [The Converse to the DL, e.g. {\cite[Lemma 4]{AAV16}}]\label{lemma:DL-conv} Let $A_1, \cdots A_g$ be projections. Then, if $\ket{\psi}$ is a state satisfying for some $\epsilon\in [0, 1]$,
    \begin{equation}
        \big\|\prod_{\chi\in [g]}  A_i \ket{\psi}\big\|^2    \leq 1-\epsilon \;,
    \end{equation}
    then there exists a $\chi \in [g]$ such that $ \|A_\chi\ket{\psi}\|^2\leq 1-\frac{\epsilon}{4g}$.
\end{lemma}

For simplicity of presentation, in this section we fix our attention to local Hamiltonians which are expectations over projections $H = \frac{1}{m}\sum_i \Pi_i$; where further we assume each projection does not commute with at most $(g-1)$ other projections. By \cref{theorem:equitable_coloring} this implies that the local terms of $H$ can be efficiently (and equitably) partitioned into $g$ commuting layers/colors (\cref{definition:layered}). We express the terms in each layer $\chi\in [g]$ as:
\begin{equation}
    H^\chi = \mathbb{E}_{i\in [m_\chi]}\Pi_i^{\chi}
\end{equation}
and associate a ``weight'' $w_\chi = m_\chi/m =\Omega(g^{-1})$ to be the number of clauses labeled the color $\chi.$ For any integer $t$, we consider the amplified Hamiltonian defined by tensor products of the DL operators applied to each layer: 
\begin{equation}
    H^{(t)} = \mathbb{I} - \mathbb{E}_\chi \bigotimes_j^t\Big(\prod_{i\in [m_\chi]} (\mathbb{I}-\Pi_i^\chi)\Big) \;,\label{eq:amp-H-DL}
\end{equation}
where as before, $ \mathbb{E}_\chi$ indicates a sample $\chi$ from the distribution $(w_1, w_2, \cdots w_g)$ over $[g]$.
The amplified Hamiltonian is Hermitian because the projections appearing in the product $\prod_{i\in [m_\chi]} (\mathbb{I}-\Pi_i^\chi)$ all commute, since they belong to the same layer.

\begin{theorem}
    [Amplification from the DL] \label{theorem:amp-from-dl} Assume that the minimum eigenvalue of $H$ is $\lambda_{\mathsf{min}}(H)\geq \gamma$. Then, the minimum eigenvalue of the amplified Hamiltonian $H^{(t)}$ from \autoref{eq:amp-H-DL} satisfies
    \begin{equation}
        \lambda_{\mathsf{min}}(H^{(t)}) \geq \min\bigg( \Theta\big(g^{-2}\big), \quad \Omega\bigg( \frac{t\cdot m\cdot \gamma}{g^4}\bigg)\bigg)\;.
    \end{equation}
\end{theorem}

\begin{proof}
    From \cref{lemma:DL}, we have that for any state $\ket{\psi}\in \mathbb{C}^{2^n}$ on a single copy of the system,
    \begin{align*}
         \bigg\|\prod_{\chi\in [g]}\prod_{i\in [m_\chi]} \bigg(\mathbb{I}-\Pi_{i}^\chi\bigg)\ket{\psi}\bigg\|^2 \leq (1+m\gamma/g^2)^{-1}\;.
    \end{align*}
    This can be amplified by taking tensor products since the operator norm is sub-multiplicative. That is, for any $\ket{\psi}\in (\mathbb{C}^{2^n})^{\ot t}$ on $t$ copies of the original system:
    \begin{align*}
         \bigg\|\bigotimes_{j=1}^t\prod_{\chi\in [g]}\prod_{i\in [m_\chi]} \bigg(\mathbb{I}-\Pi_{i}^\chi\bigg)_{\mathsf{reg}[j]}\ket{\psi}\bigg\|^2 \leq (1+m\gamma/g^2)^{-t}\;.
    \end{align*}
    Next, we observe that the order of the tensor product can be exchanged with that of the product over colors. As a consequence, we can introduce the operators $A_\chi = \otimes_j^t\prod_{i\in [m_\chi]} \big(\mathbb{I}-\Pi_{i}^\chi\big)_{\mathsf{reg}[j]}$, and re-express the operator above as
    \begin{equation*}
        \bigotimes_{j=1}^t\prod_{\chi\in [g]}\prod_{i\in [m_\chi]} \bigg(\mathbb{I}-\Pi_{i}^\chi\bigg)_{\mathsf{reg}[j]} = \prod_{\chi\in [g]} A_\chi \;.
    \end{equation*}
    Each $A_\chi$ is a product of \emph{commuting} projections; hence, $A_\chi$ itself is also a projection.
    We can therefore apply \cref{lemma:DL-conv}, which implies that there must exist one layer labeled by index $\chi$ with
    \begin{equation*}
        \bigg\|\bigotimes_j^t\prod_{i\in [m_\chi]} \bigg(\mathbb{I}-\Pi_{i}^\chi\bigg)_{\mathsf{reg}[j]}\ket{\psi}\bigg\|^2 \leq 1-\frac{1}{4g}\bigg(1-\frac{1}{(1+m\gamma/g^2)^{t}}\bigg)\;.
    \end{equation*}
    To conclude, we obtain that the energy of $\ket{\psi}$ under the amplified Hamiltonian is at least
    \begin{align*}
        \bra{\psi}H^{(t)}\ket{\psi} &= 1-  \mathbb{E}_\chi\bra{\psi}\bigotimes_j^t\prod_{i\in [m_\chi]} (\mathbb{I}-\Pi_i^\chi)\ket{\psi}\\
        &\geq \frac{\min_\chi w_\chi}{ 4g}\cdot \bigg(1-\frac{1}{(1+m\gamma/g^2)^{t}}\bigg)\;.
    \end{align*}
    The claimed result then follows by considering the regimes of $m\gamma t/g^2 \leq 0.1$ and $\geq 0.1$ separately.
    
\end{proof}

\end{appendix}

\end{document}